\documentclass[12pt]{article}
\usepackage[latin1]{inputenc}
\usepackage[british]{babel}
\usepackage{cmap}
\usepackage{lmodern}

\usepackage{amssymb, amsmath, amsthm}
\usepackage[a4paper,top=25mm,bottom=25mm,left=25mm,right=25mm]{geometry}
\usepackage{ragged2e}

\usepackage{authblk} 
\usepackage{pifont}
\usepackage{graphicx}
\usepackage[usenames,dvipsnames,svgnames,table]{xcolor}
\usepackage[figuresright]{rotating}
\usepackage{xtab} 
\usepackage{longtable} 
\usepackage{multirow}
\usepackage{footnote}
\usepackage[stable]{footmisc}
\usepackage{chngpage} 
\usepackage{pdflscape} 
\usepackage[nottoc,notlot,notlof]{tocbibind} 

\usepackage{pgfplots}
\pgfplotsset{compat=1.14}
\pgfplotsset{every tick label/.append style={font=\footnotesize}}
\usepackage{setspace}

\makesavenoteenv{tabular}
\usepackage{tabularx}
\usepackage{booktabs}
\usepackage{threeparttable} 
\usepackage[referable]{threeparttablex} 
\newcolumntype{R}{>{\raggedleft\arraybackslash}X}
\newcolumntype{L}{>{\raggedright\arraybackslash}X}
\newcolumntype{C}{>{\centering\arraybackslash}X}
\newcolumntype{A}{>{\columncolor{gray!25}}C}
\newcolumntype{a}{>{\columncolor{gray!25}}c}

\newlength{\tablen}

\usepackage{dcolumn} 
\newcolumntype{.}{D{.}{.}{-1}}

\usepackage{tikz}
\usetikzlibrary{arrows, calc, matrix, patterns, positioning, shapes, trees}
\usepackage[semicolon]{natbib}
\usepackage[hyphens]{url}
\usepackage{hyperref} 
\hypersetup{
  colorlinks   = true,    
  urlcolor     = blue,    
  linkcolor    = blue,    
  citecolor    = ForestGreen      
}
\usepackage{microtype}
\usepackage[singlelinecheck=false,justification=centering]{caption} 

\usepackage{lineno}

\usepackage[labelformat=simple]{subcaption}

\DeclareCaptionLabelFormat{parenthesis}{(#2)}
\makeatletter
\renewcommand\p@subfigure{\arabic{figure}.}
\makeatother

\DeclareCaptionLabelFormat{parenthesis}{(#2)}
\makeatletter
\renewcommand\p@subtable{\arabic{table}.}
\makeatother

\usepackage{enumitem}
\setlist[itemize]{leftmargin=2.5\parindent}
\setlist[enumerate]{leftmargin=2.5\parindent}

\def\addlegendimage{\csname pgfplots@addlegendimage\endcsname}

\theoremstyle{plain}

\newtheorem{proposition}{Proposition}[section]

\theoremstyle{definition}

\newtheorem{example}{Example}

\theoremstyle{remark}

\def\keywords{\vspace{.5em} 
{\noindent \textit{Keywords}: }}

\def\JEL{\vspace{.5em} 
{\noindent \textbf{\emph{JEL} classification number}: }}

\def\AMS{\vspace{.5em} 
{\noindent \textbf{\emph{MSC} class}: }}

\title{The effects of draw restrictions \\ on knockout tournaments}
\author{\href{https://sites.google.com/view/laszlocsato}{L\'aszl\'o Csat\'o}\thanks{~E-mail: \emph{laszlo.csato@sztaki.hu}} }
\affil{Institute for Computer Science and Control (SZTAKI) \\
E\"otv\"os Lor\'and Research Network (ELKH) \\
Laboratory on Engineering and Management Intelligence \\
Research Group of Operations Research and Decision Systems}
\affil{Corvinus University of Budapest (BCE) \\
Department of Operations Research and Actuarial Sciences}
\affil{Budapest, Hungary}
\date{\today}

\def\Dedication{
{\noindent
$\mathfrak{Wir}$ $\mathfrak{haben}$ $\mathfrak{diese}$ $\mathfrak{einfachen}$ $\mathfrak{Vorstellungen}$ $\mathfrak{an}$ $\mathfrak{die}$ $\mathfrak{Wirklichkeit}$ $\mathfrak{angekn\ddot{u}pft}$ $\mathfrak{und}$ $\mathfrak{so}$ $\mathfrak{den}$ $\mathfrak{Weg}$ $\mathfrak{gezeigt}$, $\mathfrak{wie}$ $\mathfrak{man}$ $\mathfrak{aus}$ $\mathfrak{der}$ $\mathfrak{Wirklichkeit}$ $\mathfrak{zu}$ $\mathfrak{jenen}$ $\mathfrak{einfachen}$ $\mathfrak{Vorstellungen}$ $\mathfrak{wieder}$ $\mathfrak{zur\ddot{u}ckgelangen}$ $\mathfrak{und}$ $\mathfrak{also}$ $\mathfrak{festen}$ $\mathfrak{Grund}$ $\mathfrak{gewinnen}$ $\mathfrak{kann}$, $\mathfrak{damit}$ $\mathfrak{man}$ $\mathfrak{nicht}$ $\mathfrak{gen\ddot{o}tigt}$ $\mathfrak{sei}$, $\mathfrak{im}$ $\mathfrak{R\ddot{a}sonnement}$ $\mathfrak{zu}$ $\mathfrak{St\ddot{u}tzpunkten}$ $\mathfrak{seine}$ $\mathfrak{Zuflucht}$ $\mathfrak{zu}$ $\mathfrak{nehmen,}$ $\mathfrak{die}$ $\mathfrak{selbst}$ $\mathfrak{in}$ $\mathfrak{der}$ $\mathfrak{Luft}$ $\mathfrak{schweben}$.\footnote{~
``\emph{We have connected these simple ideas with reality, and therefore shown the way by which we may return again from the reality to those simple ideas, and obtain firm ground, and not be forced in reasoning to take refuge on points of support which themselves vanish in the air.}'' (Source: Carl von Clausewitz: \emph{On War}, Book 6, Chapter 8 [Methods of Resistance]. Translated by Colonel James John Graham, London, N. Tr\"ubner, 1873. \url{http://clausewitz.com/readings/OnWar1873/TOC.htm})}
}
\vspace{0.25cm}

\flushright
\noindent (Carl von Clausewitz: \emph{Vom Kriege})

\vspace{1cm} 
\justify }

\begin{document}

\newgeometry{top=10mm,bottom=15mm,left=25mm,right=25mm}

\maketitle
\thispagestyle{empty}
\Dedication

\begin{abstract}
\noindent
The paper analyses how draw constraints influence the outcome of a knockout tournament. The research question is inspired by European club football competitions, where the organiser generally imposes an association constraint in the first round of the knockout phase: teams from the same country cannot be drawn against each other. Its effects are explored in both theoretical and simulation models. An association constraint in the first round(s) is found to increase the likelihood of same nation matchups to approximately the same extent in each subsequent round. If the favourite teams are concentrated in some associations, they will have a higher probability to win the tournament under this policy but the increase is less than linear if it is used in more rounds.
Our results might explain the recent introduction of the association constraint for both the knockout round play-offs with 16 teams and the Round of 16 in the UEFA Europa League and UEFA Europa Conference League.


\AMS{62F07, 90-10, 90B90}

\JEL{C44, C63, Z20}
\end{abstract}

\clearpage
\restoregeometry

\section{Introduction} \label{Sec1}

The Operational Research community has recently made a substantial effort to understand the unforeseen and probably unintended consequences of various sports rules and rule changes \citep{Csato2021a, KendallLenten2017, LentenKendall2021, Wright2014}. The present paper examines how imposing draw constraints affect the outcome of a knockout tournament. Our study is inspired by the club competitions of the Union of European Football Associations (UEFA). In these tournaments, more than one team can participate from certain countries. Therefore, UEFA generally imposes an association constraint both in the group stage and the first round of the subsequent knockout phase, that is, teams from the same country cannot be drawn against each other.

Unfortunately, it remains unknown why UEFA aims to avoid games between teams from the same country \citep[Footnote~33]{BoczonWilson2018}. The organiser likely has an underlying preference to maintain the international character of the competitions in their early stages since this restriction is not applied in later rounds. It is also possible that the matches played by teams from different nations attract more viewers because teams from the same country play against each other in their domestic league, too. While this lack of information on the reality of the decision-making process undertaken by UEFA and on the arguments for the association constraint does not allow for a proper cost-benefit analysis, the restrictions clearly have some potential disadvantages which are worth exploring.

\begin{table}[ht!]
  \centering
  \caption{The national association of the UEFA Europa League \\ finalists between the 2009/10 and 2020/21 seasons}
  \label{Table1}
    \rowcolors{1}{gray!20}{}
    \begin{tabularx}{0.6\textwidth}{lLL} \toprule \hiderowcolors
    \multirow{2}{*}{Season} & \multicolumn{2}{c}{Association of the} \\
          & winner & runner-up \\ \bottomrule \showrowcolors
    2009/10 & Spain & England \\
    2010/11 & Portugal & Portugal \\
    2011/12 & Spain & Spain \\
    2012/13 & England & Portugal \\
    2013/14 & Spain & Portugal \\
    2014/15 & Spain & Ukraine \\
    2015/16 & Spain & England \\
    2016/17 & England & Netherlands \\
    2017/18 & Spain & France \\
    2018/19 & England & England \\
    2019/20 & Spain & Italy \\
    2020/21 & Spain & England \\ \bottomrule
    \end{tabularx}
\end{table}

The secondary club football tournament of Europe, the UEFA Europa League, has recently been dominated by English and Spanish teams, see Table~\ref{Table1}. Among the 12 winners, three have come from England and eight from Spain (more than 90\%), while there have been seven English and nine Spanish finalists out of 24 (two-thirds). A question arises naturally: To what extent does the association constraint contribute to the probability that the winner comes from a country with several representatives?

\begin{table}[t!]
  \centering
  \caption{Same nation matchups in the UEFA Europa League \\ knockout stage between the 2009/10 and 2020/21 seasons}
  \label{Table2}
    \rowcolors{1}{}{gray!20}
    \begin{tabularx}{0.6\textwidth}{lLL} \toprule
    Season & Round & Association \\ \bottomrule
    2009/10 & quarterfinals & Spain \\ \hline
    2010/11 & semifinals & Portugal \\
    2010/11 & final & Portugal \\ \hline
    2011/12 & semifinals & Spain \\
    2011/12 & final & Spain \\ \hline
    2012/13 & \multicolumn{2}{c}{---} \\ \hline
    2013/14 & Round of 16 & Italy \\
    2013/14 & Round of 16 & Spain \\
    2013/14 & semifinals & Spain \\ \hline
    2014/15 & Round of 16 & Italy \\
    2014/15 & Round of 16 & Spain \\ \hline
    2015/16 & Round of 16 & England \\
    2015/16 & Round of 16 & Spain \\
    2015/16 & quarterfinals & Spain \\ \hline
    2016/17 & Round of 16 & Belgium \\
    2016/17 & Round of 16 & Germany \\ \hline
    2017/18 & \multicolumn{2}{c}{---} \\ \hline
    2018/19 & quarterfinals & Spain \\
    2018/19 & final & England \\ \hline
    2019/20 & \multicolumn{2}{c}{---} \\ \hline
    2020/21 & \multicolumn{2}{c}{---} \\ \toprule
    \end{tabularx}
\end{table}

In addition, as Table~\ref{Table2} reveals, some clashes have taken place between teams from the same association in the rounds without the restriction (Round of 16, quarterfinals, semifinals). This observation inspires the second issue that our study wants to address: What is the effect of the association constraint in the Round of 32 on the likelihood of a same nation matchup in the subsequent rounds?

Academic researchers have extensively discussed draw constraints in the group stage of sports tournaments. The FIFA World Cup draw has been demonstrated to be unevenly distributed due to geographical restrictions \citep{Jones1990, RathgeberRathgeber2007, Guyon2015a}.
Several papers have made suggestions to create more balanced groups in the FIFA World Cup \citep{Guyon2015a, LalienaLopez2019, CeaDuranGuajardoSureSiebertZamorano2020}.
Draw constraints can be a powerful tool to avoid matches where a team has misaligned incentives \citep{Csato2022a}.
\citet{Guyon2022a} proposes a novel format for hybrid tournaments consisting of a preliminary group stage followed by a knockout phase and adapts it to the constraints put by the UEFA on the draw.

On the other hand, draw restrictions in knockout tournaments have received far less attention. According to \citet{Kiesl2013} and \citet{KlossnerBecker2013}, the mechanism used in the UEFA Champions League Round of 16 draw implies that every result of the draw could not have the same probability. However, the UEFA draw procedure remains close to a constrained-best in terms of fairness \citep{BoczonWilson2022}. 

Our main contributions to the topic can be summarised as follows:
\begin{itemize}
\item
First in the literature, the effects of using draw constraints in a knockout tournament are investigated in a theoretical model.
\item
Introducing an association constraint in the first round(s) is documented to increase the likelihood of same nation matchups to \emph{approximately the same extent} in each subsequent round.
\item
We show that, if the favourite teams are concentrated in some national associations, they will have a higher probability to win the tournament in the presence of an association constraint. This is especially important for international tournaments where higher-ranked countries can delegate more teams such as in the competitions organised by the UEFA. However, the increase is \emph{less than linear} if the restriction is imposed in more rounds.
\item
According to simulations based on empirical data, the UEFA policy of banning the same nation matchups in the first round of the knockout stage (Round of 32) in the UEFA Europa League is a strange compromise between maintaining the international character of the tournament and avoiding the dominance of certain nations. Imposing the association constraint in later rounds is an option worth further consideration.
\end{itemize}
While some of the above results might seem obvious at first sight, we think the insights written in italics are far from trivial and can be important for tournament design.

The simulation model is based on the format of the UEFA Europa League used between the 2009/10 and 2020/21 seasons, when the knockout phase started with the Round of 32. Currently, the Europa League---and the UEFA Europa Conference League, the third tier European competition launched in the 2021/22 season---contains knockout round play-offs with 16 teams, followed by the Round of 16, quarterfinals, semifinals, and the final. 
Our last finding says that preventing teams from the same national association to meet against each other only in the Round of 32 can hardly be justified. Since now the association constraint is used in the draw of both the knockout round play-offs and the Round of 16, the results can explain this recent decision of the UEFA.

The paper has some connections to \citet{BoczonWilson2022} who analyse the UEFA Champions League Round of 16 draw in order to quantify how the draw mechanism affects expected assignments and alters expected tournament prizes. A former version of the paper, \citet{BoczonWilson2018} also show in Figure~10 that allowing one same nation pairing in the draw decreases the aggregated number of same nation matchups in the subsequent rounds by about 10\%.
However, our investigation has important novelties since \citet{BoczonWilson2022} neither attempt to formulate a mathematical model for analysing the effects of restrictions imposed in different rounds nor consider their implications for the probability that the winner comes from a given national association. Therefore, while the choice of UEFA administrators for the draw mechanism is verified by \citet{BoczonWilson2022}, the current work sheds much more light on the alternatives to the current structure of draw constraints. To summarise, tournament organisers will definitely not understand the effects of draw restrictions used in knockout tournaments from \citet{BoczonWilson2022}---but they can get insight into this issue from our study.

The remainder of the paper is laid out as follows.
Section~\ref{Sec2} outlines the design of the UEFA Europa League between the 2009/10 and 2020/21 seasons and presents the draw policies to be compared. A basic mathematical model is provided in Section~\ref{Sec3}. Section~\ref{Sec4} investigates the consequences of the association constraint in 12 recent Europa League seasons via simulations. In particular, the methodology is detailed in Section~\ref{Sec41} and the results are discussed in Section~\ref{Sec42}. Section~\ref{Sec5} provides concluding remarks and some directions for future research. 

\section{The rules of the UEFA Europa League} \label{Sec2}

Between the 2009/10 and 2020/21 seasons, the group stage of the UEFA Europa League has consisted of 12 groups with four teams each. The top two from each group---altogether 24 clubs---have qualified for the knockout stage, where eight third-placed teams from the UEFA Champions League (the most prestigious club football tournament of the continent) group stage have joined them.
In this period, the design of the knockout phase has not changed. The Round of 32 pairings have been determined by a draw according to the following constraints \citep[Article~17]{UEFA2020b}:
\begin{itemize}
\item
The 12 Europa League group winners and the four best third-placed teams from the Champions League group stage are drawn against the 12 Europa League group runners-up and the remaining four third-placed teams from the Champions League group stage;
\item
Clubs from the same national association cannot play against each other;
\item
The winners and runners-up of the same Europa League group cannot play against each other.
\end{itemize}

In the Round of 32 draw, the group constraint can hardly be debated; repeated matchups are worth avoiding as they are probably less interesting for the spectators. The association constraint might be explained by a similar argument because these teams play against each other in their domestic leagues, too. But the knockout phase of the parallel UEFA Champions League starts with the Round of 16, where both the group and association constraints are in use. Hence, the association constraint would be reasonable to require in the Europa League Round of 16 draw.

To understand the effects of the association constraint, three alternatives to draw a knockout tournament with 32 teams are considered:
\begin{itemize}
\item
\emph{Method $\emptyset$}: there are no restrictions in the draw due to the associations of the teams.
\item
\emph{Method 32}: two teams from the same association cannot be drawn against each other in the Round of 32.
\item
\emph{Method 16}: two teams from the same association can be drawn against each other neither in the Round of 32 nor in the Round of 16.
\end{itemize}
The three options will be evaluated first in a simple mathematical model, followed by a simulation based on the historical results of the UEFA Europa League.

Even though the Europa League is a less prestigious competition compared to the UEFA Champions League, we think it provides a better starting point to study the role of draw restrictions because a knockout tournament with 32 teams offers more opportunities to use an association constraint. In particular, the rules of the Champions League allow since the 2015/16 season \citep[Article~3.08]{UEFA2015a} that five teams from the same country play in the quarterfinals, the second round of the knockout phase, when the association constraint cannot be satisfied. Furthermore, even if a valid assignment exists under the association constraint, this restriction may severely reduce the number of feasible pairings. For example, three English and three Spanish teams played in the quarterfinals of the 2021/22 Champions League, hence, the association constraint would have reduced the number of solutions by 60\% from $7 \times 5 \times 3 = 105$ to $42$.\footnote{~The number of cases when two English teams play against each other is $3 \times 5 \times 3 = 45$ as they come from a set of three teams and the remaining six teams should be matched. Analogously, there are $45$ cases when two Spanish teams are paired. The sum is $90$, however, $3 \times 3 \times 3 = 27$ among them contain matches between both English and Spanish teams. Thus, the number of valid assignments is $105 - 2 \times 45 + 27 = 42$.}

\section{A basis mathematical model} \label{Sec3}

Some consequences of the draw constraints can be uncovered in a simple probabilistic framework.
Assume that there are two teams from the same country $W$ in a knockout tournament with $32$ teams, each of them having a probability of $w$ to advance against any of the remaining 30 teams.
Inspired by Section~\ref{Sec1}, two measures will be analysed:
\begin{itemize}
\item
The probability that the two teams from country $W$ play against each other;
\item
The probability that the winner comes from country $W$.
\end{itemize}
Thus, it is sufficient to distinguish only two types of teams, teams from country $W$ and teams from outside $W$. Consequently, the winning probabilities against the same type are irrelevant with respect to both metrics.

While using a uniform $w$ for both teams of country $W$ is a simplification, we can capture with it an important aspect of the problem without the need of investigating a complex interaction of several variables. Note that introducing different probabilities $w_1$ and $w_2$ for winning against the remaining 30 teams implies three types of teams. Hence, the outcome of a match between the two teams from country $W$ should be a third parameter, or it should be modelled as a function of $w_1$ and $w_2$, which can be controversial.
Furthermore, it will be seen in Section~\ref{Sec4} that the results of the simulation analysis reinforce those of the mathematical model, and might justify the assumption of a unique probability $w$.
Finally, analytical findings seem to be difficult to derive even under this restrictive condition if there are three or more teams from country $W$ as the number of their possible pairings to be accounted for increases rapidly.

Consider Method $\emptyset$.
The probability of a same nation matchup in the Round of 32 is:
\[
P^{(\emptyset)}_{32} = \frac{1}{31}.
\]
In the Round of 16, it is:
\[
P^{(\emptyset)}_{16} = \left( 1 - P^{(\emptyset)}_{32} \right) \cdot \frac{1}{15} \cdot w^2 = \frac{30}{31} \cdot \frac{1}{15} \cdot w^2.
\]
In the quarterfinals, it is:
\[
P^{(\emptyset)}_{8} = \frac{30}{31} \cdot \frac{14}{15} \cdot \frac{1}{7} \cdot w^4.
\]
In the semifinals, it is:
\[
P^{(\emptyset)}_{4} = \frac{30}{31} \cdot \frac{14}{15} \cdot \frac{6}{7} \cdot \frac{1}{3} \cdot w^6.
\]
In the final, it is:
\[
P^{(\emptyset)}_{2} = \frac{30}{31} \cdot \frac{14}{15} \cdot \frac{6}{7} \cdot \frac{2}{3} \cdot w^8.
\]

Consider Method 32.
The probability of a same nation matchup in the Round of 32 is $P^{(32)}_{32} = 0$.
In the Round of 16, it is:
\[
P^{(32)}_{16} = \frac{1}{15} \cdot w^2.
\]
In the quarterfinals, it is:
\[
P^{(32)}_{8} = \frac{14}{15} \cdot \frac{1}{7} \cdot w^4.
\]
In the semifinals, it is:
\[
P^{(32)}_{4} = \frac{14}{15} \cdot \frac{6}{7} \cdot \frac{1}{3} \cdot w^6.
\]
In the final, it is:
\[
P^{(32)}_{2} = \frac{14}{15} \cdot \frac{6}{7} \cdot \frac{2}{3} \cdot w^8.
\]

Consider Method 16.
The probability of a same nation matchup in the Round of 32 is $P^{(16)}_{32} = 0$.
In the Round of 16, it is $P^3_{16} = 0$, too.
In the quarterfinals, it is:
\[
P^{(16)}_{8} = \frac{1}{7} \cdot w^4.
\]
In the semifinals, it is:
\[
P^{(16)}_{4} = \frac{6}{7} \cdot \frac{1}{3} \cdot w^6.
\]
In the final, it is:
\[
P^{(16)}_{2} = \frac{6}{7} \cdot \frac{2}{3} \cdot w^8.
\]

Now we turn to the second issue.
Consider Method $\emptyset$.
A team from country $W$ can win the tournament in two ways. First, the two teams from country $W$ might play a match against each other, hence the probability that one of them will be the final winner is:
\[
P^{(\emptyset)}_{32} \cdot w^4 + P^{(\emptyset)}_{16} \cdot w^3 + P^{(\emptyset)}_{8} \cdot w^2 + P^{(\emptyset)}_{4} \cdot w + P^{(\emptyset)}_{2}.
\]
Second, the other team from country $W$ might be eliminated before it plays against the winner, which has the following probability:
\[
2 w^5 \cdot \left[ \frac{30}{31} \left( 1-w \right) + \frac{30}{31} \cdot \frac{14}{15} w \left( 1-w \right) + \frac{30}{31} \cdot \frac{14}{15} \cdot \frac{6}{7} w^2 \left( 1-w \right) + \frac{30}{31} \cdot \frac{14}{15} \cdot \frac{6}{7} \cdot \frac{2}{3} w^3 \left( 1-w \right) \right].
\]
The sum of the two terms provides the chance that the winner comes from association $W$.

Consider Method 32.
The following formula gives the probability that the winner is from association $W$ and plays against the other team from its association:
\[
P^{(32)}_{32} \cdot w^4 + P^{(32)}_{16} \cdot w^3 + P^{(32)}_{8} \cdot w^2 + P^{(32)}_{4} \cdot w + P^{(32)}_{2}.
\]
Analogously, the probability that the winner comes from association $W$ and does not play against the other team from its association in the tournament is:
\[
2 w^5 \cdot \left[ \left( 1-w \right) + \frac{14}{15} w \left( 1-w \right) + \frac{14}{15} \cdot \frac{6}{7} w^2 \left( 1-w \right) + \frac{14}{15} \cdot \frac{6}{7} \cdot \frac{2}{3} w^3 \left( 1-w \right) \right].
\]

Consider Method 16.
The following formula gives the probability that the winner is from association $W$ and plays against the other team from its association:
\[
P^{(16)}_{32} \cdot w^4 + P^{(16)}_{16} \cdot w^3 + P^{(16)}_{8} \cdot w^2 + P^{(16)}_{4} \cdot w + P^{(16)}_{2}.
\]
Similarly, the probability that the winner comes from association $W$ and does not play against the other team from its association in the tournament is:
\[
2 w^5 \cdot \left[ \left( 1-w \right) + w \left( 1-w \right) + \frac{6}{7} w^2 \left( 1-w \right) + \frac{6}{7} \cdot \frac{2}{3} w^3 \left( 1-w \right) \right].
\]

These calculations can be easily generalised to a knockout tournament with $2^k$ teams.

If the association constraint holds only in the Round of 32 (Method 32), the probability that the two teams from country $W$ meet against each other is multiplied by the same factor in any subsequent rounds since $P^{(32)}_{16} / P^{(\emptyset)}_{16} = P^{(32)}_{8} / P^{(\emptyset)}_{8} = P^{(32)}_{4} / P^{(\emptyset)}_{4} = P^{(32)}_{2} / P^{(\emptyset)}_{2} = 31/30$. Analogously, if the association constraint holds in both the Round of 32 and the Round of 16 (Method 16), the probability of such a clash is multiplied by the same factor in any subsequent rounds since $P^{(16)}_{16} / P^{(\emptyset)}_{16} = P^{(16)}_{8} / P^{(\emptyset)}_{8} = P^{(16)}_{4} / P^{(\emptyset)}_{4} = P^{(16)}_{2} / P^{(\emptyset)}_{2} = 31/30 \cdot 15/14$. However, the sums of these probabilities are naturally not linear functions of each other because they depend on the winning probability $w$.

\begin{figure}[t!]

\begin{tikzpicture}
\begin{axis}[
name = axis1,
title = Unweighted rounds,
title style = {font=\small},
xlabel = Winning probability $w$,
x label style = {font=\small},
width = 0.46\textwidth,
height = 0.5\textwidth,
ymajorgrids = true,
xmin = 0,
xmax = 1,
ylabel = {Change in the probability of a same \\ nation matchup in percentage points},
y label style = {align = center, font=\small},
yticklabel style = {scaled ticks = false, /pgf/number format/fixed},
] 
\draw[very thick](axis cs:\pgfkeysvalueof{/pgfplots/xmin},0)  -- (axis cs:\pgfkeysvalueof{/pgfplots/xmax},0);
\addplot [red, thick] coordinates {
(0,-3.2258064516129)
(0.02,-3.22572036123522)
(0.04,-3.22546126098141)
(0.06,-3.22502664344781)
(0.08,-3.22441226193596)
(0.1,-3.22361202580645)
(0.12,-3.22261784775885)
(0.14,-3.22141943682253)
(0.16,-3.22000402906713)
(0.18,-3.21835604626571)
(0.2,-3.21645667096774)
(0.22,-3.21428332466348)
(0.24,-3.21180903394509)
(0.26,-3.20900166779448)
(0.28,-3.20582302735168)
(0.3,-3.20222776774194)
(0.32,-3.19816212976382)
(0.34,-3.19356245746497)
(0.36,-3.18835347585592)
(0.38,-3.18244630123725)
(0.4,-3.17573615483871)
(0.42,-3.16809974869395)
(0.44,-3.15939231089803)
(0.46,-3.14944421561945)
(0.48,-3.13805718146253)
(0.5,-3.125)
(0.52,-3.11000375452018)
(0.54,-3.09275648725692)
(0.56,-3.07289727159501)
(0.58,-3.05000964396765)
(0.6,-3.0236143483871)
(0.62,-2.99316134477335)
(0.64,-2.95802103047042)
(0.66,-2.91747462256349)
(0.68,-2.87070364683469)
(0.7,-2.81677847741936)
(0.72,-2.75464586944882)
(0.74,-2.68311542518989)
(0.76,-2.60084493241549)
(0.78,-2.50632451196505)
(0.8,-2.39785950967742)
(0.82,-2.27355206610323)
(0.84,-2.13128129562796)
(0.86,-1.96868200486093)
(0.88,-1.78312187836983)
(0.9,-1.57167705806451)
(0.92,-1.33110604075775)
(0.94,-1.05782181665535)
(0.96,-0.747862169751679)
(0.98,-0.396858059331118)
(1,0)
};

\addplot [blue, thick, dashdotdotted] coordinates {
(0,-3.2258064516129)
(0.02,-3.22838687539894)
(0.04,-3.23612548172597)
(0.06,-3.24901421108224)
(0.08,-3.26703941336559)
(0.1,-3.29018151152074)
(0.12,-3.31841451065344)
(0.14,-3.35170533264384)
(0.16,-3.39001295057292)
(0.18,-3.43328729156835)
(0.2,-3.48146787096774)
(0.22,-3.53448211498975)
(0.24,-3.59224332339493)
(0.26,-3.65464821791083)
(0.28,-3.72157401648756)
(0.3,-3.79287496774194)
(0.32,-3.86837827424086)
(0.34,-3.94787932756596)
(0.36,-4.03113617239403)
(0.38,-4.11786311111973)
(0.4,-4.20772335483871)
(0.42,-4.30032062080198)
(0.44,-4.39518957074366)
(0.46,-4.49178497877681)
(0.48,-4.58946951184384)
(0.5,-4.6875)
(0.52,-4.78501206810058)
(0.54,-4.88100299475439)
(0.56,-4.97431265869824)
(0.58,-5.06360242703888)
(0.6,-5.14733183410138)
(0.62,-5.22373289391434)
(0.64,-5.29078188365493)
(0.66,-5.34616842966837)
(0.68,-5.38726172196864)
(0.7,-5.41107367741936)
(0.72,-5.41421886608549)
(0.74,-5.39287100953892)
(0.76,-5.34271585419264)
(0.78,-5.25890021703052)
(0.8,-5.13597699539171)
(0.82,-4.96784592676038)
(0.84,-4.74768987880415)
(0.86,-4.46790644419584)
(0.88,-4.12003460904591)
(0.9,-3.69467625806452)
(0.92,-3.18141227386418)
(0.94,-2.56871298210646)
(0.96,-1.84384268848755)
(0.98,-0.99275804785004)
(1,0)
};
\end{axis}

\begin{axis}[
at = {(axis1.south east)},
xshift = 0.15\textwidth,
title = Weighted rounds,
title style = {font=\small},
xlabel = Winning probability $w$,
x label style = {font=\small},
width = 0.46\textwidth,
height = 0.5\textwidth,
ymajorgrids = true,
xmin = 0,
xmax = 1,
ylabel = {Change in the probability of a same \\ nation matchup in percentage points},
y label style = {align = center, font=\small},
yticklabel style = {scaled ticks = false, /pgf/number format/fixed},
legend style = {font=\small,at={(-1.5,-0.2)},anchor=north west,legend columns=2},
legend entries = {Method 32 compared to Method $\emptyset \qquad$,Method 16 compared to Method $\emptyset$}
] 
\draw[very thick](axis cs:\pgfkeysvalueof{/pgfplots/xmin},0)  -- (axis cs:\pgfkeysvalueof{/pgfplots/xmax},0);
\addplot [red, thick] coordinates {
(0,-0.3584229390681)
(0.02,-0.358401860765952)
(0.04,-0.358338014562876)
(0.06,-0.358229555776689)
(0.08,-0.358073373074841)
(0.1,-0.357865030073391)
(0.12,-0.357598677360208)
(0.14,-0.357266928727063)
(0.16,-0.35686069361945)
(0.18,-0.356368956037144)
(0.2,-0.355778488342721)
(0.22,-0.355073486659417)
(0.24,-0.354235112763923)
(0.26,-0.353240925603861)
(0.28,-0.352064183793911)
(0.3,-0.350672998668715)
(0.32,-0.349029315694878)
(0.34,-0.347087700268588)
(0.36,-0.344793902149533)
(0.38,-0.342083171006008)
(0.4,-0.338878293770268)
(0.42,-0.335087322727384)
(0.44,-0.330600961485041)
(0.46,-0.325289574195897)
(0.48,-0.318999781628312)
(0.5,-0.311550605905445)
(0.52,-0.302729123956897)
(0.54,-0.292285587951272)
(0.56,-0.279927969202196)
(0.58,-0.265315880264545)
(0.6,-0.248053828161802)
(0.62,-0.227683749909649)
(0.64,-0.203676779725088)
(0.66,-0.175424195534585)
(0.68,-0.142227490618891)
(0.7,-0.103287514456393)
(0.72,-0.0576926250510396)
(0.74,-0.00440579325507867)
(0.76,0.0577494021790276)
(0.78,0.130103951858843)
(0.8,0.214159830141666)
(0.82,0.311607713246614)
(0.84,0.42434606945809)
(0.86,0.554501691565104)
(0.88,0.704451743457102)
(0.9,0.87684739457245)
(0.92,1.07463911767184)
(0.94,1.30110372718427)
(0.96,1.55987323714967)
(0.98,1.85496561955741)
(1,2.19081754565627)
};

\addplot [blue, thick, dashdotdotted] coordinates {
(0,-0.3584229390681)
(0.02,-0.359054473906914)
(0.04,-0.360957573026084)
(0.06,-0.364157487350709)
(0.08,-0.368695522625032)
(0.1,-0.374627860727087)
(0.12,-0.38202388953131)
(0.14,-0.390964021341188)
(0.16,-0.401536974206063)
(0.18,-0.413836484728243)
(0.2,-0.427957415258577)
(0.22,-0.443991212670692)
(0.24,-0.462020670196114)
(0.26,-0.48211393809451)
(0.28,-0.504317723225334)
(0.3,-0.52864961187916)
(0.32,-0.555089444519038)
(0.34,-0.58356966537421)
(0.36,-0.613964564120567)
(0.38,-0.646078321174247)
(0.4,-0.679631762416795)
(0.42,-0.71424772346234)
(0.44,-0.749434917869271)
(0.46,-0.784570197990901)
(0.48,-0.818879091451664)
(0.5,-0.851414490527394)
(0.52,-0.88103336600026)
(0.54,-0.906371371350975)
(0.56,-0.9258151974429)
(0.58,-0.937472532144712)
(0.6,-0.939139473630312)
(0.62,-0.928265240386683)
(0.64,-0.901914015252438)
(0.66,-0.856723755101817)
(0.68,-0.788861792080919)
(0.7,-0.693977046594984)
(0.72,-0.567148666537555)
(0.74,-0.402830901544417)
(0.76,-0.194794015347127)
(0.78,0.0639389664068152)
(0.8,0.381159881208398)
(0.82,0.765547641664971)
(0.84,1.22674525415871)
(0.86,1.77544186965059)
(0.88,2.42346009780279)
(0.9,3.18384882130056)
(0.92,4.07098175296241)
(0.94,5.10066198393537)
(0.96,6.29023277698039)
(0.98,7.65869486456057)
(1,9.2268305171531)
};
\end{axis}
\end{tikzpicture}

\captionsetup{justification=centering}
\caption{The effect of draw restrictions on the probability of a match \\ between two teams from the same national association, analytical model}
\label{Fig1}

\end{figure}
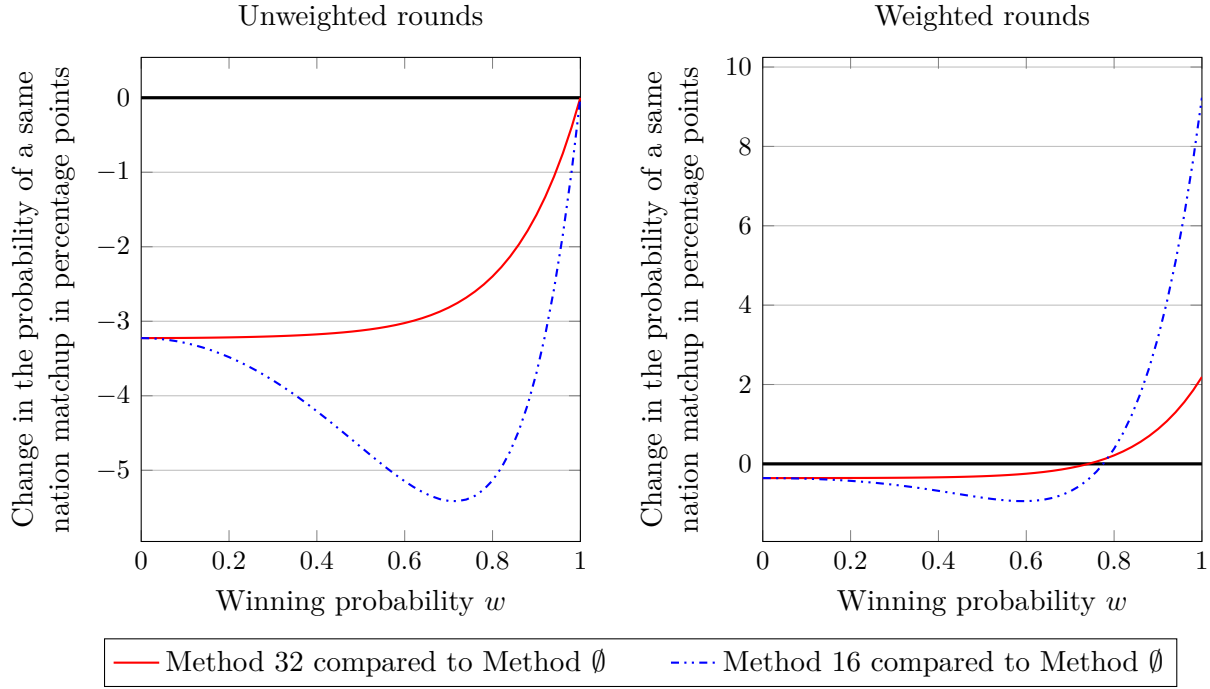


According to the left panel of Figure~\ref{Fig1}, Method 32 is able to reduce the probability of a same nation matchup by at most $1/31$ ($\approx 3.23$ percentage points), which is the chance that the two teams from country $W$ are paired in the Round of 32 in the absence of the draw constraint (Method $\emptyset$). On the other hand, Method 16 becomes the most effective when the winning probability is around 0.7, that is, association $W$ is relatively strong. The explanation is obvious: prohibiting a particular match in the Round of 16 can be useful if the two teams reach this stage with a substantial probability. The chart suggests a theoretical result, too.

\begin{proposition} \label{Prop1}
The probability that the two teams from country $W$ play against each other during the whole tournament is not increased by the draw restrictions.
\end{proposition}

\begin{proof}
For Method 32,
\begin{eqnarray*}
P^{(\emptyset)}_{32} + P^{(\emptyset)}_{16} + P^{(\emptyset)}_{8} + P^{(\emptyset)}_{4} + P^{(\emptyset)}_{2} & = & \frac{1}{31} + \frac{30}{31} \cdot \left( P^{(32)}_{16} + P^{(32)}_{8} + P^{(32)}_{4} + P^{(32)}_{2} \right) \geq \\
 & \geq & P^{(32)}_{16} + P^{(32)}_{8} + P^{(32)}_{4} + P^{(32)}_{2}
\end{eqnarray*}
because the latter sum is a probability, hence it cannot be greater than one.

The calculation for Method 16 is also elementary and left to the reader.
\end{proof}

The association constraint aims to prevent same nation matchups, probably because a match between the two teams from the same country $W$ decreases attention and diversity. However, these costs for a match between the two teams from the same country $W$ are not necessarily uniform, and they can be higher in a later round of the tournament where fewer matches are played. For this purpose, a weighted cost is computed with the weights derived from the 2019/20 UEFA club competitions revenue distribution system \citep{UEFA2019b}. Qualification for the knockout stage of the Europa League is awarded by the following amounts:
(a) 0.5 million Euros for the Round of 32;
(b) 1.1 million Euros for the Round of 16;
(c) 1.5 million Euros for the quarterfinals;
(d) 2.4 million Euros for the semifinals; and
(e) 4.5 million Euros for the final.
Since the prize money provides the weights, the weighted cost of a same nation matchup under Method $\emptyset$ is:
\[
\frac{0.5 P^{(\emptyset)}_{32} + 1.1 P^{(\emptyset)}_{16} + 1.5 P^{(\emptyset)}_{8} + 2.4 P^{(\emptyset)}_{4} + 4.5 P^{(\emptyset)}_{2}}{0.5 + 1.1 + 1.5 + 2.4 + 4.5}.
\]
This value serves mainly for comparative purposes; its maximum is less than one as the two dominating teams do not necessarily meet in the final even if $w=1$.

The right panel of Figure~\ref{Fig1} shows the effects of restrictions on the weighted cost of a match between the two teams from country $W$. While Methods 32 and 16 are better than Method $\emptyset$ if $w$ does not exceed 0.75, the prohibition of such a clash in the first round(s) becomes detrimental in the case of strong teams.

\begin{figure}[t!]

\begin{tikzpicture}
\begin{axis}[
name = axis1,
xlabel = Winning probability $w$,
x label style = {font=\small},
width = 0.97\textwidth,
height = 0.6\textwidth,
ymajorgrids = true,
xmin = 0,
xmax = 1,
ylabel = {Change in the probability of winning \\ the tournament in percentage points},
y label style = {align = center, font=\small},
yticklabel style = {scaled ticks = false, /pgf/number format/fixed, /pgf/number format/precision=3},
legend style = {font=\small,at={(-0.05,-0.15)},anchor=north west,legend columns=2},
legend entries = {Method 32 compared to Method $\emptyset \qquad$,Method 16 compared to Method $\emptyset$}
] 
\draw[very thick](axis cs:\pgfkeysvalueof{/pgfplots/xmin},0)  -- (axis cs:\pgfkeysvalueof{/pgfplots/xmax},0);
\addplot [red, thick] coordinates {
(0,0)
(0.02,-4.94795700686452E-07)
(0.04,-7.57539875146322E-06)
(0.06,-3.66224862888465E-05)
(0.08,-0.000110284160749172)
(0.1,-0.000255917419354839)
(0.12,-0.000503032979889384)
(0.14,-0.00088074756489216)
(0.16,-0.00141525030357587)
(0.18,-0.00212729310723568)
(0.2,-0.00302971870967741)
(0.22,-0.00412504453924402)
(0.24,-0.00540312570336417)
(0.26,-0.00683892612018577)
(0.28,-0.00839043322478591)
(0.3,-0.00999675870967734)
(0.32,-0.0115764754308433)
(0.34,-0.0130262489213473)
(0.36,-0.0142198309046654)
(0.38,-0.0150074917892914)
(0.4,-0.0152159793548393)
(0.42,-0.0146491017078691)
(0.44,-0.0130890440929364)
(0.46,-0.010298541290904)
(0.48,-0.00602404012247718)
(0.5,0)
(0.52,0.00804550646490987)
(0.54,0.0183837160866646)
(0.56,0.0312764289096115)
(0.58,0.0469653851539076)
(0.6,0.0656595406451604)
(0.62,0.0875199988085579)
(0.64,0.112642339408234)
(0.66,0.141036065672942)
(0.68,0.172600872270351)
(0.7,0.207099416774192)
(0.72,0.244126256811261)
(0.74,0.283072593978223)
(0.76,0.323086443882742)
(0.78,0.363027829287621)
(0.8,0.401418570322576)
(0.82,0.436386222074125)
(0.84,0.465601685570982)
(0.86,0.486209993250186)
(0.88,0.494753744416809)
(0.9,0.487088640000011)
(0.92,0.458290539057005)
(0.94,0.40255343198754)
(0.96,0.313077697292097)
(0.98,0.181947979939268)
(1,0)
};

\addplot [blue, thick, dashdotdotted] coordinates {
(0,0)
(0.02,-0.00000002041905152)
(0.04,-6.24154378240004E-07)
(0.06,-4.51756249673143E-06)
(0.08,-0.00001810144165888)
(0.1,-5.23885714285716E-05)
(0.12,-0.000123270398326492)
(0.14,-0.000251141950832642)
(0.16,-0.000459899729346558)
(0.18,-0.000775334394593276)
(0.2,-0.00122294857142857)
(0.22,-0.00182523999404033)
(0.24,-0.00259850154316363)
(0.26,-0.00354920246612546)
(0.28,-0.00467002922631171)
(0.3,-0.00593568000000011)
(0.32,-0.00729852382543878)
(0.34,-0.00868425381155401)
(0.36,-0.00998768363175955)
(0.38,-0.0110688577620026)
(0.4,-0.0117496685714282)
(0.42,-0.0118111974388532)
(0.44,-0.0109920225486435)
(0.46,-0.00898776291557435)
(0.48,-0.00545215649976963)
(0.5,-6.93889390390723E-16)
(0.52,0.00778769094372661)
(0.54,0.018355909849041)
(0.56,0.0321648505284183)
(0.58,0.0496756006122284)
(0.6,0.0713318399999996)
(0.62,0.0975370085618016)
(0.64,0.128626367775372)
(0.66,0.164833339932957)
(0.68,0.206249466084341)
(0.7,0.252777279999999)
(0.72,0.304075350139205)
(0.74,0.35949469489463)
(0.76,0.418005728254645)
(0.78,0.478114843479732)
(0.8,0.537769691428569)
(0.82,0.594252157792696)
(0.84,0.644057989707203)
(0.86,0.682761966996859)
(0.88,0.704867456694347)
(0.9,0.703639131428568)
(0.92,0.670917572826235)
(0.94,0.596914420201033)
(0.96,0.469986662518251)
(0.98,0.276388607922662)
(1,0)
};
\end{axis}
\end{tikzpicture}

\captionsetup{justification=centering}
\caption{The effect of draw restrictions on the probability that a team \\ from the association with two teams wins the tournament, analytical model}
\label{Fig2}

\end{figure}
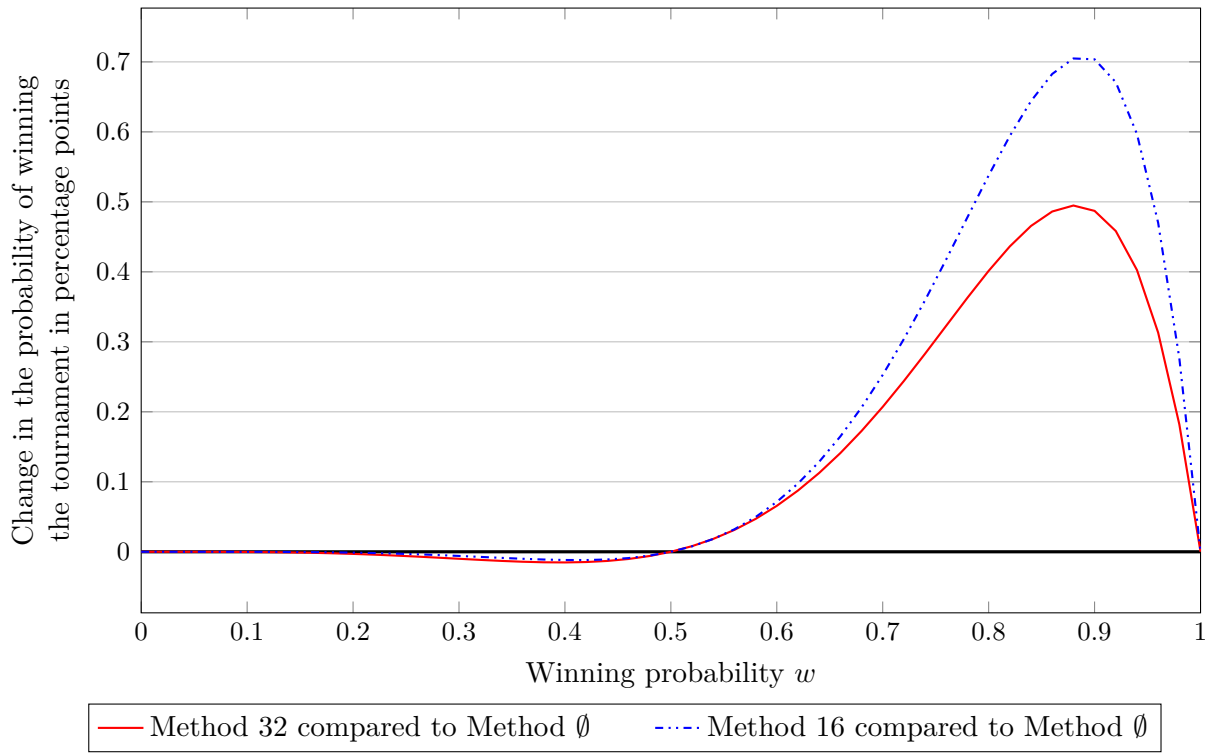


Turning to the second question of our study, Figure~\ref{Fig2} reveals how the impact of the draw constraint depends on the winning probability $w$ with respect to the likelihood that the winner of the tournament comes from country $W$. Obviously, if the teams of this association are weak ($w < 0.5$), both Methods 32 and 16 are unfavourable for them as they cannot play against each other at the beginning of the tournament. On the other hand, the effects are more serious---although still moderated in absolute terms---when they often defeat any other team. Again, the result is intuitive; banning a match between the strong teams helps them to win the tournament. It is also worth noting that imposing the constraint in the first two rounds rather than only in the first round is far from doubling the effect.

To summarise, the abstract mathematical model conveys two important messages:
\begin{itemize}
\item
If draw constraints apply in the first round(s), the probability of a clash between the two teams from the same association is multiplied by the same factor in each subsequent round.
\item
If the association with two teams is relatively strong, imposing an association constraint favours its teams. On the other hand, excluding same nation matchups in the first two rounds is substantially less than doubling the effect of using this restriction only in the first round.
\end{itemize}
The second observation is especially relevant because the winning probability $w$ is more likely to be above $0.5$ in the real world. The reason is that the teams qualifying for similar knockout tournaments are not chosen randomly, countries with better teams usually get more places in the group stage and their clubs also have a higher probability to qualify for the knockout stage.

\section{Simulation results} \label{Sec4}

In the following, the effects of the association constraint in the UEFA Europa League will be quantified. As Section~\ref{Sec3} has demonstrated, an exact mathematical computation remains impossible due to the complex interactions between the constraints for different national associations. 
Therefore, the research questions will be addressed via Monte Carlo simulations, a standard approach in the analysis of tournament designs \citep{ScarfYusofBilbao2009, GoossensBelienSpieksma2012, LasekGagolewski2018, Csato2021b}. To that end, the draw procedure should be replicated and it needs to be determined which team advances to the next stage from a given match. We will also try to connect the results of the simulations to the analytical findings.

\subsection{Methodology} \label{Sec41}

If there are some draw constraints, it is a non-trivial problem to pair the competitors since a valid matching should be obtained. UEFA has followed its usual procedure \citep{BoczonWilson2022, Guyon2014c, Kiesl2013, KlossnerBecker2013} in the Europa League Round of 32 draw between the 2009/10 and the 2020/21 seasons.
In particular, the 16 unseeded teams (12 Europa League group runners-up and four lower-ranked third-placed teams from the Champions League) are drawn randomly from an urn. For each team, the set of possible opponents is established by a computer program in order to avoid any dead-end, a situation when the remaining teams cannot be paired. Then a club is drawn randomly from this set of possible opponents, and the pair of the two drawn teams is added to the matching.
The video of the 2020/21 UEFA Europa League Round of 32 draw is available at \url{https://www.uefa.com/uefaeuropaleague/draws/2021/2001241/}.

\begin{example} \label{Examp1}
Assume that there are four seeded teams $T1$--$T4$ and four unseeded teams $T5$--$T8$ still to be drawn under the following constraints:
\begin{itemize}
\item
Certain pairs of clubs have played in the same group, hence they cannot be drawn against each other ($T1$--$T5$, $T2$--$T6$, $T3$--$T7$, $T4$--$T8$);
\item
Some clubs are from the same national association: teams $\{ T2, \, T5 \}$ from country $A$ and teams $\{ T1, \, T4, \, T6 \}$ from country $B$.
\end{itemize}
First, $T5$ is drawn. It has two possible opponents, $T3$ and $T4$, because $T1$ is excluded by the group constraint and $T2$ is excluded by the association constraint. However, if $T5$ plays against $T3$, there remains no feasible assignment for $T6$ ($T2$ is prohibited by the group constraint, whereas $T1$ and $T4$ are excluded by the association constraint). Hence, $T5$ should be drawn against $T4$ since a draw condition is anticipated to apply.
Second, $T6$ is drawn and paired with $T3$.
Finally, there are two group winners and two group runners-up without any restriction, thus the opponent of $T7$ will be either $T1$ (implying a match between $T8$ and $T2$) or $T2$ (implying a match between $T8$ and $T1$).
\end{example}

According to Example~\ref{Examp1}, the draw in the Round of 32 is more complicated than it might seem at first glance. For instance, the conditions affecting the clubs still to be drawn should be taken into account, and the cardinality of the set of teams against which an unseeded team is allowed to play is not necessarily decreasing.

Our simulation is based on the UEFA mechanism above if there exists any draw restriction, that is, in the Round of 32 under Method 32, as well as in the Round of 32 and in the Round of 16 under Method 16. In all other rounds, a random matching is picked up due to the absence of draw constraints.

The UEFA emergency panel ruled on 17 July 2014 that Ukrainian and Russian clubs could not be drawn against each other due to the political unrest between the countries. This constraint is never taken into account in our study.

Regarding the likelihood of advancing, we use separate assumptions for the two issues addressed:
\begin{itemize}
\item
Same nation matchup: the probability of winning depends on the past performance of the club. Four cases are distinguished according to the rules of the Round of 32 draw (Europa League group winner, Europa League runner-up, seeded Champions League third-placed team, unseeded Champions League third-placed team).
\item
The winner comes from a particular country: the probability of winning depends on the national association of the club. Three cases are distinguished as suggested by Table~\ref{Table1} (English clubs, Spanish clubs, clubs from all other countries).
\end{itemize}

\begin{table}[t!]
  \centering
  \caption{Winning probabilities in the UEFA Europa League \\ knockout stage between the 2009/10 and 2020/21 seasons}
  \label{Table3}

\begin{threeparttable}
\addtocounter{table}{-1} 
  \begin{subtable}{\textwidth}
  \centering
  \caption{By type of the qualification}
  \label{Table3a}
  \rowcolors{1}{}{gray!20}
    \begin{tabularx}{\textwidth}{LLc} \toprule
    Team 1 & Team 2 & Winning probability of Team 1 \\ \bottomrule
    EL group winner & EL group runner-up & $85/139 \approx 0.612$ \\
    EL group winner & CL seeded & $22/43 \approx 0.512$ \\
    EL group winner & CL unseeded & $30/59 \approx 0.508$ \\ \hline
    EL group runner-up & CL seeded & $14/53 \approx 0.264$ \\
    EL group runner-up & CL unseeded & $5/11 \approx 0.455$ \\ \hline
    CL seeded & CL unseeded & $11/16 = 0.6875$ \\ \bottomrule
    \end{tabularx}
  \end{subtable}
\begin{tablenotes} \footnotesize
\item
\emph{Abbreviations}: EL group winner = UEFA Europa League group winner; \\
EL group runner-up = UEFA Europa League group runner-up; \\
CL seeded = higher-ranked third-placed team from the UEFA Champions League group stage; \\
CL unseeded = lower-ranked third-placed team from the UEFA Champions League group stage.
\end{tablenotes}
\end{threeparttable}

\vspace{0.5cm}  
  \begin{subtable}{\textwidth}
  \centering
  \caption{By national association}
  \label{Table3b}
  \rowcolors{1}{}{gray!20}
    \begin{tabularx}{0.8\textwidth}{LLc} \toprule
    Team 1 & Team 2 & Winning probability of Team 1 \\ \bottomrule
    English & Spanish & $6/16 \approx 0.375$ \\
    English & Other & $50/67 \approx 0.746$ \\
    Spanish & Other & $58/73 \approx 0.795$ \\ \toprule
    \end{tabularx}
  \end{subtable}
\end{table}

In both cases, the winning probability is determined by historical data from the 12 Europa League seasons organised between 2009/10 and 2020/21. They are presented in Table~\ref{Table3}. For example, Europa League group winners have played 139 matches against Europa League group runners-up, and the former team has won 85 times (Table~\ref{Table3a}). Analogously, there have been 67 clashes between clubs from England and all other nations except for Spain, among which 50 have been won by the English club (Table~\ref{Table3b}). Interestingly, the data align with the assumption of strong stochastic transitivity of the pairwise winning probability matrix, that is, if a team $x$ is stronger than another team $y$, then $x$ defeats any third team with a higher probability than $y$. This property is often used in the probabilistic analysis of knockout tournaments \citep{Arlegi2022, ArlegiDimitrov2020, Hwang1982, HorenRiezman1985, Schwenk2000}.
If the two clubs have the same type, the probability of winning is assumed to be 0.5.

\begin{sidewaystable}
  \centering
  \caption{National associations with more than one club in the UEFA Europa League knockout stage by season}
  \label{Table4}
\begin{threeparttable}
\rowcolors{1}{}{gray!20}
    \begin{tabularx}{1\textwidth}{l CCCC CCCC CCCC} \toprule
    & 2009/10 & 2010/11 & 2011/12 & 2012/13 & 2013/14 & 2014/15 & 2015/16 & 2016/17 & 2017/18 & 2018/19 & 2019/20 & 2020/21 \\ \bottomrule
    AUT   &       &       &       &       &       &       &       &       &       & $1|1|0|0$ & $1|0|1|0$ & $0|1|0|1$ \\
    BEL   & $1|1|0|1$ &       & $3|0|0|0$ &       &       & $1|0|0|1$ &       & $1|2|0|0$ &       & $1|0|0|1$ & $1|0|0|1$ & $0|1|1|0$ \\
    CZE   &       &       &       & $1|1|0|0$ & $0|1|0|1$ &       &       &       &       & $0|1|0|1$ &       &  \\
    ENG   & $0|2|0|1$ & $2|0|0|0$ & $0|1|2|0$ & $1|2|1|0$ & $1|1|0|0$ & $1|1|0|1$ & $2|0|1|0$ & $0|1|1|0$ &       & $2|0|0|0$ & $2|1|0|0$ & $3|0|1|0$ \\
    FRA   & $0|1|1|0$ & $1|1|0|0$ &       &       &       &       & $0|2|0|0$ & $1|0|1|0$ & $0|3|0|0$ &       &       &  \\
    GER   & $1|2|1|0$ & $2|0|0|0$ & $1|1|0|0$ & $1|3|0|0$ &       & $1|1|0|0$ & $1|2|1|0$ & $1|0|0|1$ & $0|0|1|1$ & $2|0|0|0$ & $0|2|0|1$ & $2|0|0|0$ \\
    GRE   &       & $0|2|0|0$ & $1|0|1|0$ &       &       &       &       & $0|2|0|0$ &       &       &       &  \\
    ITA   & $1|0|1|0$ &       & $0|2|0|0$ & $1|2|0|0$ & $1|1|1|1$ & $3|1|0|1$ & $2|1|0|0$ & $2|0|0|0$ & $3|0|0|1$ & $0|1|2|0$ & $0|1|1|0$ & $3|0|0|0$ \\
    NED   & $1|2|0|0$ & $1|0|2|0$ & $2|1|0|1$ &       & $1|0|0|1$ & $1|1|0|1$ &       & $1|1|0|0$ &       &       & $0|1|1|0$ & $1|0|1|0$ \\
    POL   &       &       & $0|2|0|0$ &       &       &       &       &       &       &       &       &  \\
    POR   & $2|0|0|0$ & $2|0|1|1$ & $1|1|0|1$ &       & $0|0|1|1$ &       & $1|1|1|0$ &       & $1|0|1|0$ & $0|1|1|0$ & $2|1|1|0$ & $0|2|0|0$ \\
    ROU   &       &       &       & $1|0|1|0$ &       &       &       &       &       &       &       &  \\
    RUS   &       & $2|0|1|1$ & $0|2|0|0$ & $1|1|0|1$ & $1|1|0|0$ & $1|0|1|0$ & $2|0|0|0$ & $1|1|0|1$ & $2|0|1|1$ & $1|1|0|0$ &       &  \\
    SCO   &       &       &       &       &       &       &       &       &       &       & $1|1|0|0$ &  \\
    SRB   &       &       &       &       &       &       &       &       & $0|2|0|0$ &       &       &  \\
    ESP   & $1|2|0|1$ & $1|1|0|0$ & $2|0|1|0$ & $0|2|0|0$ & $2|1|0|0$ & $0|2|1|0$ & $1|1|0|2$ & $0|3|0|0$ & $2|1|1|0$ & $3|0|1|0$ & $2|1|0|0$ & $1|2|0|0$ \\
    SUI   &       & $0|1|0|1$ &       &       &       &       &       &       &       &       &       &  \\
    TUR   & $2|0|0|0$ &       & $1|0|0|1$ &       &       & $1|1|0|0$ & $1|1|0|1$ & $2|0|1|0$ &       & $0|1|0|1$ &       &  \\
    UKR   &       & $1|1|0|0$ &       & $2|0|0|1$ & $0|3|1|0$ & $1|1|0|0$ &       &       &       & $1|0|0|1$ &       & $0|0|1|1$ \\ \toprule
    \end{tabularx}

\begin{tablenotes} \footnotesize
\item
\emph{Abbreviations:} AUT = Austria; BEL = Belgium; CZE = Czech Republic; ENG = England; FRA = France; GER = Germany; GRE = Greece; ITA = Italy; NED = Netherlands; POL = Poland; POR = Portugal; ROU = Romania, RUS = Russia; SCO = Scotland; SRB = Serbia; ESP = Spain; SUI = Switzerland; TUR = Turkey; UKR = Ukraine.
\item
\emph{Note:} In any tuple $i|j|k|\ell$, $i$ indicates the number of Europa League group winners, $j$ the number of Europa League runners-up, $k$ the number of seeded third-placed teams from the Champions League group stage, and $\ell$ the number of unseeded third-placed teams from the Champions League group stage. The total number of exclusions generated by an association in the Round of 32 draw equals $(i + k) \times (j + \ell)$.
\end{tablenotes}
\end{threeparttable}
\end{sidewaystable}

The effect of the association constraint depends on the identity of the participants.
Table~\ref{Table4} details the association constraints in each season. Note that at least two Spanish teams have qualified for the Round of 32 every year, while the association constraint for England and Germany have not influenced the tournament in one season only, respectively.
Therefore, all simulations are run separately for the 12 seasons with 10 million iterations.

\subsection{Assessing the implications of the association constraint} \label{Sec42}

In the following, the numerical results of the simulations will be presented. They always indicate \emph{relative} changes; for instance, if the probability of a particular event under Method $\emptyset$ is 40\%, which is changed by 5\% under Method 32, then the corresponding probability of this event under Method 32 is 42\%.

\begin{figure}[t!]
\centering

\begin{tikzpicture}
\begin{axis}[
name = axis1,
width = 0.46\textwidth, 
height = 0.5\textwidth,
title = {Method 32 compared to Method $\emptyset$},
title style = {align=center, font=\small},
xmajorgrids = true,
ymajorgrids = true,
xmin = 0,
xmax = 12.5,
scaled x ticks = false,
xlabel = {Relative increase in \%},
xlabel style = {align=center, font=\small},
xticklabel style = {/pgf/number format/fixed,/pgf/number format/precision=5},
ytick style = {draw = none},
symbolic y coords = {2009/10, 2010/11, 2011/12, 2012/13, 2013/14, 2014/15, 2015/16, 2016/17, 2017/18, 2018/19, 2019/20, 2020/21, \textbf{Average}},
ytick = data,
y dir = reverse,
legend style = {font=\small,at={(0.3,-0.2)},anchor=north west,legend columns=4},
legend entries = {Round of 16$\qquad$, Quarterfinals$\qquad$, Semifinals$\qquad$, Final}
]
\addplot [brown, only marks, mark = pentagon, very thick] coordinates{
(3.30987307212618,2009/10)
(2.57040661531005,2010/11)
(2.37340882465971,2011/12)
(4.17986347179615,2012/13)
(4.6468383693036,2013/14)
(4.19085952127145,2014/15)
(3.08969758261051,2015/16)
(2.4103261538029,2016/17)
(2.37642495852981,2017/18)
(1.98738253601789,2018/19)
(3.12389050654633,2019/20)
(0.97117693648725,2020/21)
(2.9358457123718,\textbf{Average})
};
\addplot [blue, only marks, mark = square, very thick] coordinates{
(3.23497466667448,2009/10)
(2.32829197898468,2010/11)
(2.2030581034379,2011/12)
(4.09186586770096,2012/13)
(4.34384886012111,2013/14)
(4.02588854760859,2014/15)
(2.6659055122217,2015/16)
(2.24461261402835,2016/17)
(2.25470526028961,2017/18)
(1.84764235297881,2018/19)
(2.68834996098539,2019/20)
(1.18922789442708,2020/21)
(2.75986430162154,\textbf{Average})
};
\addplot [red, only marks, mark = star, thick] coordinates{
(3.22785937303458,2009/10)
(2.08137183621036,2010/11)
(1.74815798758847,2011/12)
(3.53787356343034,2012/13)
(4.43389077530905,2013/14)
(3.73415249325657,2014/15)
(2.27971647696374,2015/16)
(2.09617241834517,2016/17)
(2.04939124336587,2017/18)
(1.6023437184008,2018/19)
(2.03609546321286,2019/20)
(0.966571480249279,2020/21)
(2.48279973578058,\textbf{Average})
};
\addplot [ForestGreen, only marks, mark = triangle, very thick] coordinates{
(2.92957392302011,2009/10)
(1.87396005339711,2010/11)
(1.41402786373799,2011/12)
(3.45867296843874,2012/13)
(4.24722137951956,2013/14)
(3.89738768904244,2014/15)
(1.71654340094347,2015/16)
(1.49547572302657,2016/17)
(1.88552034973637,2017/18)
(0.979469019302814,2018/19)
(2.06140051432602,2019/20)
(0.672300996201725,2020/21)
(2.21929615672443,\textbf{Average})
};
\end{axis}

\begin{axis}[
at = {(axis1.south east)},
xshift = 0.15\textwidth,
width = 0.46\textwidth, 
height = 0.5\textwidth,
title = {Method 16 compared to Method $\emptyset$},
title style = {align=center, font=\small},
xmajorgrids = true,
ymajorgrids = true,
xmin = 0,
xmax = 12.5,
scaled x ticks = false,
xlabel = {Relative increase in \%},
xlabel style = {align=center, font=\small},
xticklabel style = {/pgf/number format/fixed,/pgf/number format/precision=5},
ytick style = {draw = none},
symbolic y coords = {2009/10, 2010/11, 2011/12, 2012/13, 2013/14, 2014/15, 2015/16, 2016/17, 2017/18, 2018/19, 2019/20, 2020/21, \textbf{Average}},
ytick = data,
y dir = reverse,
]
\addplot [blue, only marks, mark = square, very thick] coordinates{
(9.71097112059602,2009/10)
(9.03429021680095,2010/11)
(8.87301825201379,2011/12)
(10.6818724868012,2012/13)
(10.953091463159,2013/14)
(10.3520948847847,2014/15)
(9.30226824340694,2015/16)
(9.33076939806217,2016/17)
(8.54346972803239,2017/18)
(8.58866054301202,2018/19)
(9.46347954520121,2019/20)
(7.80971153610561,2020/21)
(9.38697478483133,\textbf{Average})
};
\addplot [red, only marks, mark = star, thick] coordinates{
(9.53549860700644,2009/10)
(8.28162821152278,2010/11)
(8.16637263651647,2011/12)
(10.299907926779,2012/13)
(10.6962914972263,2013/14)
(9.50038081257685,2014/15)
(8.3493033808997,2015/16)
(9.23080129786213,2016/17)
(7.88878468582777,2017/18)
(8.07698530303977,2018/19)
(8.72435221145269,2019/20)
(7.38325047912933,2020/21)
(8.84446308748661,\textbf{Average})
};
\addplot [ForestGreen, only marks, mark = triangle, very thick] coordinates{
(8.62500720344088,2009/10)
(7.28807232650686,2010/11)
(7.46390723876167,2011/12)
(9.23979290741395,2012/13)
(10.264118333839,2013/14)
(8.84217106251424,2014/15)
(7.59609517998781,2015/16)
(8.22620567964907,2016/17)
(6.97462928885322,2017/18)
(7.19744266814049,2018/19)
(8.19371429871874,2019/20)
(6.64528730349685,2020/21)
(8.04637029094355,\textbf{Average})
};
\end{axis}
\end{tikzpicture}

\caption{The effect of draw restrictions on the probability of a match between \\ two teams from the same national association, UEFA Europa League \\ by season, winning probabilities derived from empirical data}
\label{Fig3}

\end{figure}
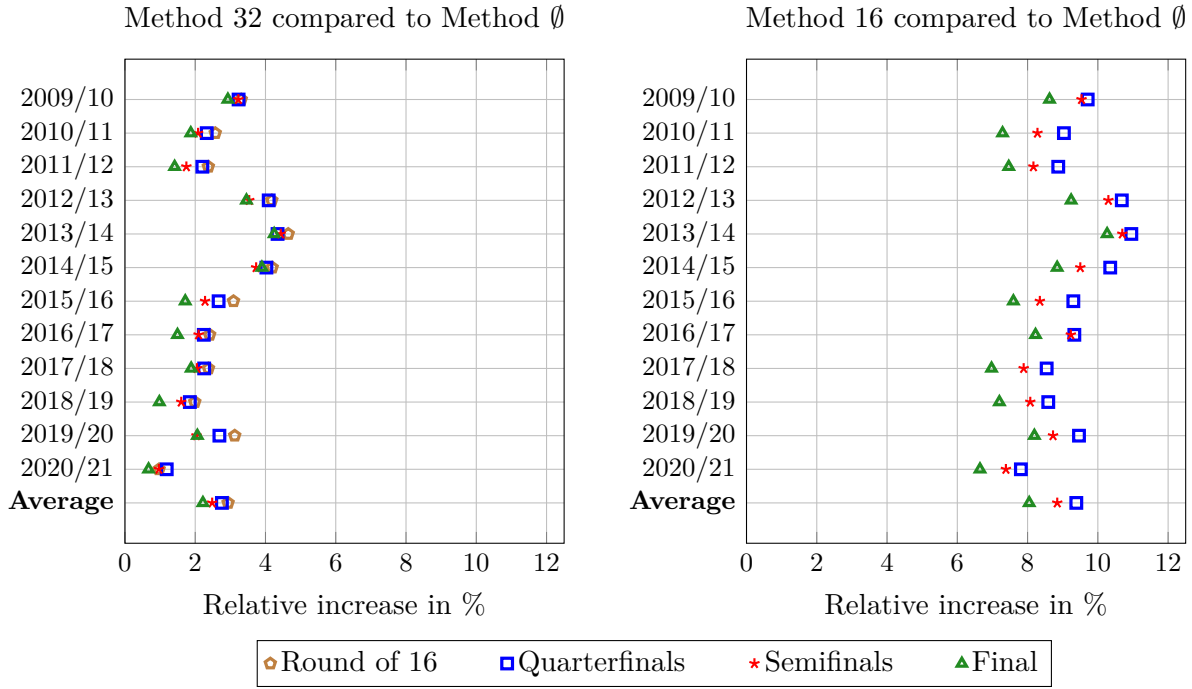


Figure~\ref{Fig3} shows how the association constraint influences the frequency of a match played by clubs from the same country. The UEFA policy has increased the probability of such a clash in each round of the Europa League other than the Round of 32 by about 2-3\% and at most 5\% in the period considered. Remarkably, the changes are mostly driven by the distribution of the teams between the national associations in the given season, that is, they almost coincide for the Round of 16, quarterfinals, semifinals, and the final, which reinforces the finding from the basic mathematical model. This observation remains valid if the association constraint is introduced in the Round of 16, too, when the probability of same nation matchups grows by at least 6\% but not more than 12\% in every subsequent round.

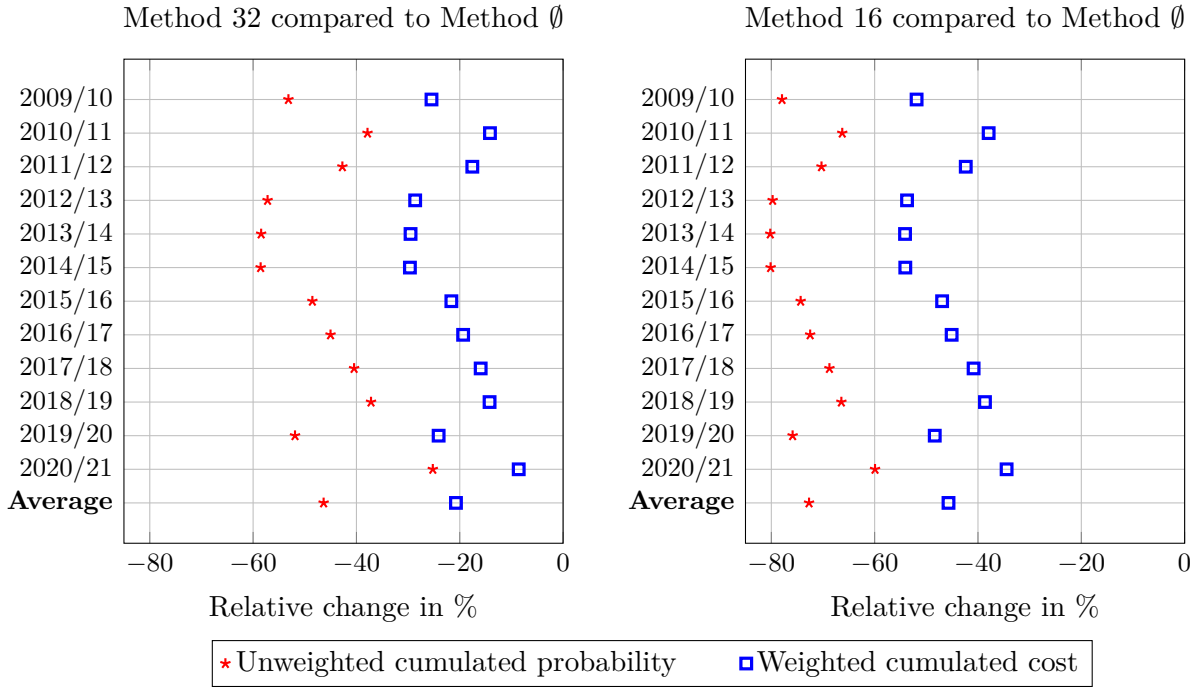
\begin{figure}[t!]
\centering

\begin{tikzpicture}
\begin{axis}[
name = axis3,
width = 0.46\textwidth, 
height = 0.5\textwidth,
title = {Method 32 compared to Method $\emptyset$},
title style = {align=center, font=\small},
xmajorgrids = true,
ymajorgrids = true,
xmin = -85,
xmax = 0,
scaled x ticks = false,
xlabel = {Relative change in \%},
xlabel style = {align=center, font=\small},
xticklabel style = {/pgf/number format/fixed},
ytick style = {draw = none},
symbolic y coords = {2009/10, 2010/11, 2011/12, 2012/13, 2013/14, 2014/15, 2015/16, 2016/17, 2017/18, 2018/19, 2019/20, 2020/21, \textbf{Average}},
ytick = data,
y dir = reverse,
legend style = {font=\small,at={(0.2,-0.2)},anchor=north west,legend columns=2},
legend entries = {Unweighted cumulated probability$\qquad$, Weighted cumulated cost}
]
\addplot [red, only marks, mark = star, thick] coordinates{
(-53.1639719350959,2009/10)
(-37.8582139656493,2010/11)
(-42.7213467321635,2011/12)
(-57.201197078578,2012/13)
(-58.4445791893669,2013/14)
(-58.5425095144154,2014/15)
(-48.5305479699887,2015/16)
(-45.0142501035427,2016/17)
(-40.4532770947181,2017/18)
(-37.1780329624039,2018/19)
(-51.8898903966319,2019/20)
(-25.2194943925949,2020/21)
(-46.3514426112624,\textbf{Average})
};
\addplot [blue, only marks, mark = square, very thick] coordinates{
(-25.4721575496414,2009/10)
(-14.1551791370807,2010/11)
(-17.5702312880404,2011/12)
(-28.6327894347644,2012/13)
(-29.5258965675039,2013/14)
(-29.6489238418811,2014/15)
(-21.6247793528488,2015/16)
(-19.3487760038739,2016/17)
(-15.9529174061667,2017/18)
(-14.2315556358594,2018/19)
(-24.0900523784225,2019/20)
(-8.59073828518992,2020/21)
(-20.7369997401061,\textbf{Average})
};
\end{axis}

\begin{axis}[
at = {(axis3.south east)},
xshift = 0.15\textwidth,
width = 0.46\textwidth, 
height = 0.5\textwidth,
title = {Method 16 compared to Method $\emptyset$},
title style = {align=center, font=\small},
xmajorgrids = true,
ymajorgrids = true,
xmin = -85,
xmax = 0,
scaled x ticks = false,
xlabel = {Relative change in \%},
xlabel style = {align=center, font=\small},
xticklabel style = {/pgf/number format/fixed},
ytick style = {draw = none},
symbolic y coords = {2009/10, 2010/11, 2011/12, 2012/13, 2013/14, 2014/15, 2015/16, 2016/17, 2017/18, 2018/19, 2019/20, 2020/21, \textbf{Average}},
ytick = data,
y dir = reverse,
]
\addplot [red, only marks, mark = star, thick] coordinates{
(-77.9317612207183,2009/10)
(-66.2708840669298,2010/11)
(-70.2884330248107,2011/12)
(-79.7456037960648,2012/13)
(-80.2025047091472,2013/14)
(-80.1366275797335,2014/15)
(-74.3223241292392,2015/16)
(-72.4828839692633,2016/17)
(-68.7530844524033,2017/18)
(-66.4426553542642,2018/19)
(-75.8774508055901,2019/20)
(-59.9276143528294,2020/21)
(-72.6984856217495,\textbf{Average})
};
\addplot [blue, only marks, mark = square, very thick] coordinates{
(-51.880486755066,2009/10)
(-37.9348330573902,2010/11)
(-42.3637216258231,2011/12)
(-53.73082099331,2012/13)
(-54.1183735477582,2013/14)
(-54.0671327847386,2014/15)
(-46.9419452561796,2015/16)
(-45.0957507240464,2016/17)
(-40.8466032994541,2017/18)
(-38.6481161232352,2018/19)
(-48.3668272367797,2019/20)
(-34.4592011050575,2020/21)
(-45.7044843757366,\textbf{Average})
};
\end{axis}
\end{tikzpicture}

\caption{The effect of draw restrictions on the cumulated probability of a match \\ between two teams from the same national association, UEFA Europa League \\ by season, winning probabilities derived from empirical data}
\label{Fig4}

\end{figure}


Imposing the association constraint raises the likelihood of a game played by two teams from the same country in all subsequent rounds.
Thus, Figure~\ref{Fig4} plots the cumulated impact both in the unweighted and weighted settings; in the latter case, a match at the end of the tournament is punished more strongly, similar to Section~\ref{Sec3}. Although UEFA has reduced the probability of such an unwanted matchup by at most 60\%, the average gain is only 20\% and the decrease always remains below 30\% if the rounds of these clashes are taken into account. On the other hand, extending the association constraint to the Round of 16 would have cut the chance of a same nation matchup by at least 60\% in the unweighted, and by more than 30\% even in the weighted scenario such that the expected reduction is about 45\% even in the latter case. Shortly, the effectiveness of Method 16 in the worst case almost coincides with the effectiveness of Method 32 in the best case.

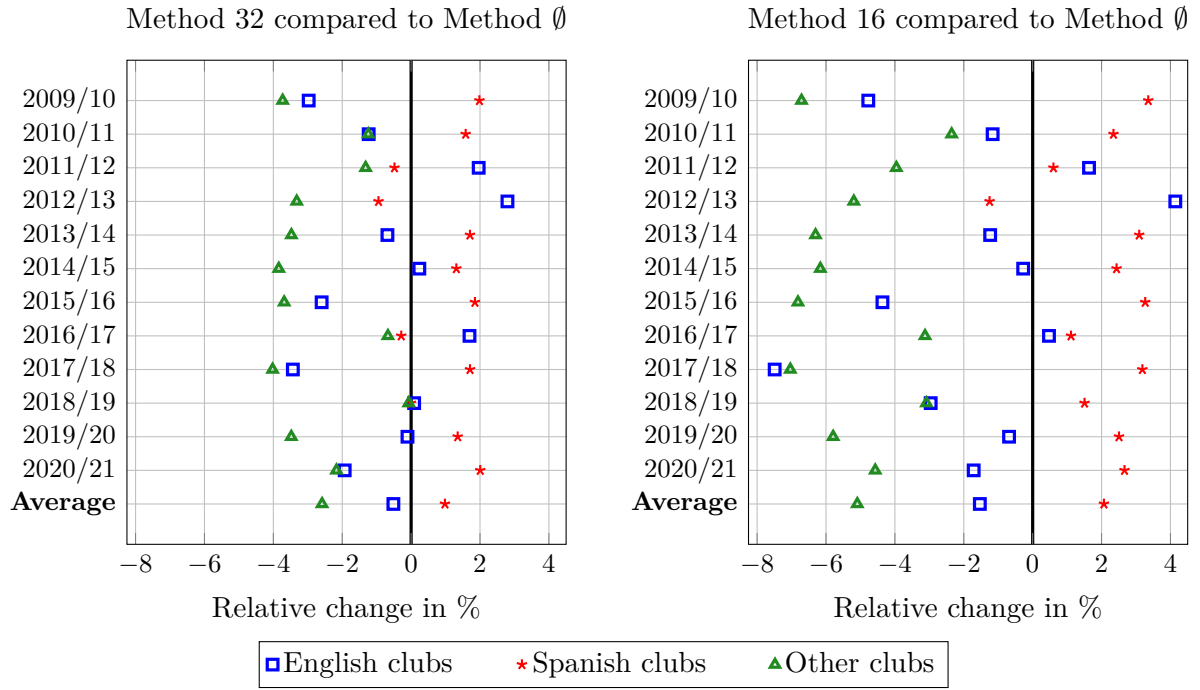
\begin{figure}[t!]
\centering

\begin{tikzpicture}
\begin{axis}[
name = axis1,
width = 0.46\textwidth, 
height = 0.5\textwidth,
title = {Method 32 compared to Method $\emptyset$},
title style = {font = \small},
xmajorgrids,
ymajorgrids,
xmin = -8.25,
xmax = 4.5,
scaled x ticks = false,
xlabel = {Relative change in \%},
xlabel style = {align=center, font=\small},
xticklabel style = {/pgf/number format/fixed,/pgf/number format/precision=5},
ytick style = {draw = none},
symbolic y coords = {2009/10, 2010/11, 2011/12, 2012/13, 2013/14, 2014/15, 2015/16, 2016/17, 2017/18, 2018/19, 2019/20, 2020/21, \textbf{Average}},
ytick = data,
y dir = reverse,
legend style = {font=\small,at={(0.3,-0.2)},anchor=north west,legend columns=4},
legend entries = {English clubs$\qquad$, Spanish clubs$\qquad$, Other clubs},
extra x ticks = 0,
extra x tick labels = ,
extra x tick style = {grid = major, major grid style = {black,very thick}},
]
\addplot [blue, only marks, mark = square, very thick] coordinates{
(-2.97372078915096,2009/10)
(-1.2316830791463,2010/11)
(1.96138086317623,2011/12)
(2.7968420045313,2012/13)
(-0.687808870989137,2013/14)
(0.238211014255496,2014/15)
(-2.59778871978111,2015/16)
(1.6946533835523,2016/17)
(-3.43337345598131,2017/18)
(0.0875162002278218,2018/19)
(-0.105629659939188,2019/20)
(-1.9290129255138,2020/21)
(-0.51503450289655,\textbf{Average})
};
\addplot [red, only marks, mark = star, thick] coordinates{
(1.9831994930865,2009/10)
(1.5835375722965,2010/11)
(-0.482828581631634,2011/12)
(-0.947261968233104,2012/13)
(1.70808494839947,2013/14)
(1.31338618860177,2014/15)
(1.85257077242984,2015/16)
(-0.286398154214951,2016/17)
(1.70969141415249,2017/18)
(-0.00410468935291286,2018/19)
(1.35206378575612,2019/20)
(2.00678677628128,2020/21)
(0.982393963130934,\textbf{Average})
};
\addplot [ForestGreen, only marks, mark = triangle, very thick] coordinates{
(-3.72911932721705,2009/10)
(-1.24300698142337,2010/11)
(-1.32063449688351,2011/12)
(-3.32057336441451,2012/13)
(-3.47554675013403,2013/14)
(-3.84206666442454,2014/15)
(-3.68543936742947,2015/16)
(-0.673660247759877,2016/17)
(-4.02092501089557,2017/18)
(-0.0693270397473023,2018/19)
(-3.47853958419441,2019/20)
(-2.1673527676063,2020/21)
(-2.58551596684415,\textbf{Average})
};
\end{axis}

\begin{axis}[
at = {(axis1.south east)},
xshift = 0.15\textwidth,
width = 0.46\textwidth, 
height = 0.5\textwidth,
title = {Method 16 compared to Method $\emptyset$},
title style = {font = \small},
xmajorgrids,
ymajorgrids,
xmin = -8.25,
xmax = 4.5,
scaled x ticks = false,
xlabel = {Relative change in \%},
xlabel style = {align=center, font=\small},
xticklabel style = {/pgf/number format/fixed,/pgf/number format/precision=5},
ytick style = {draw = none},
symbolic y coords = {2009/10, 2010/11, 2011/12, 2012/13, 2013/14, 2014/15, 2015/16, 2016/17, 2017/18, 2018/19, 2019/20, 2020/21, \textbf{Average}},
ytick = data,
y dir = reverse,
extra x ticks = 0,
extra x tick labels = ,
extra x tick style = {grid = major, major grid style = {black,very thick}},
]

\addplot [blue, only marks, mark = square, very thick] coordinates{
(-4.77256684692118,2009/10)
(-1.16245796242358,2010/11)
(1.63583243314691,2011/12)
(4.14091168095554,2012/13)
(-1.23431022129503,2013/14)
(-0.275647392126344,2014/15)
(-4.36156828984045,2015/16)
(0.473653702544641,2016/17)
(-7.49035057718872,2017/18)
(-2.96083300339909,2018/19)
(-0.68533752404677,2019/20)
(-1.71277110258983,2020/21)
(-1.53378709193199,\textbf{Average})
};
\addplot [red, only marks, mark = star, thick] coordinates{
(3.35520126162856,2009/10)
(2.34575875021708,2010/11)
(0.602457616014829,2011/12)
(-1.2512863104623,2012/13)
(3.09357252610021,2013/14)
(2.43506693973683,2014/15)
(3.26731176535862,2015/16)
(1.11264953545394,2016/17)
(3.18449757316046,2017/18)
(1.50591165150529,2018/19)
(2.50623030722694,2019/20)
(2.66579374796971,2020/21)
(2.06859711365919,\textbf{Average})
};
\addplot [ForestGreen, only marks, mark = triangle, very thick] coordinates{
(-6.71077722184704,2009/10)
(-2.35160413086217,2010/11)
(-3.95741934003732,2011/12)
(-5.19272032959649,2012/13)
(-6.30403731088118,2013/14)
(-6.16442378788419,2014/15)
(-6.81377370487356,2015/16)
(-3.12496258630804,2016/17)
(-7.03596365568896,2017/18)
(-3.08490449828349,2018/19)
(-5.79261686255783,2019/20)
(-4.57100037467216,2020/21)
(-5.09201698362437,\textbf{Average})
};
\end{axis}
\end{tikzpicture}

\caption{The effect of draw restrictions on the probability that clubs from \\ some national associations win the tournament, UEFA Europa League \\ by season, winning probabilities derived from empirical data}
\label{Fig5}

\end{figure}


The introduction of the association constraint may contribute to the dominance of countries with outstanding teams. According to Table~\ref{Table3}, Spanish clubs progress against any non-English team with a probability of 80\%, which explains their excellent performance in the Europa League (see Table~\ref{Table1}). As Figure~\ref{Fig5} presents, the UEFA rule has only marginally increased the already high winning probability of Spanish teams, which has been above 40\% in each season. Nonetheless, other teams have lost each year because of prohibited clashes in the Round of 32. Imposing the association constraint in the Round of 16 (Method 16), too, would have decreased further the chances of teams outside England and Spain, but the effects seem to remain at a tolerable level. The simulation reinforces our result obtained from the mathematical model, namely, the association constraint favours the teams of relatively strong countries with respect to their winning probabilities. However, the increase is less than linear if the restriction is imposed in more rounds.

Consequently, the rules of the draw in the UEFA Europa League knockout stage can hardly be justified: requiring the association constraint only in the Round of 32 is a strange compromise between avoiding same nation matchups and the dominance of English and Spanish clubs. Banning such clashes in the Round of 16 would have significantly reduced the likelihood of a match played by clubs from the same country at a moderate price (although this is mostly paid by the clubs outside England and Spain). UEFA is encouraged to investigate the issue more deeply and consider applying the association constraint in later rounds.

\section{Conclusions} \label{Sec5}

Inspired by the recent seasons of the UEFA Europa League, the present paper has attempted to uncover some consequences of applying draw constraints in a knockout tournament. First, we have formulated a simple mathematical model to understand how such a restriction can affect the probability of a match played by two teams from the same country, as well as the likelihood that the winner comes from a particular country. After that, the role of the association constraint has been analysed via Monte Carlo simulations based on historical results of the Europa League.

The main findings of the study can be summarised in the following way:
\begin{itemize}
\item
Imposing the association constraint in the Round of 32 draw has increased the chance of a same nation matchup by about 2-3\% and at most by 5\% in any subsequent round. Extending the restriction to the Round of 16 draw would have led to a rise of 6--12\%.
\item
Imposing the association constraint in the Round of 32 draw has reduced the probability of such an unwanted matchup by at most 60\%, but the average gain is only about 20\% and never more than 30\% if the higher cost of a match played in a later round is taken into account. Extending the restriction to the Round of 16 draw would have led to a decrease of at least 60\% with an average decline of 45\% in the weighted setting.
The effectiveness of Method 16 in the worst case almost coincides with the effectiveness of Method 32 in the best case.
\item
The association constraint used in the Round of 32 draw has increased the likelihood that a Spanish club wins the Europa League in most seasons, with a mean of 1\% in relative terms. The winning probability of teams outside England and Spain has declined by at most 4\%. Extending the restriction to the Round of 16 draw would have doubled these relative changes. 
\end{itemize}
The results are in line with the implications of the mathematical model, that is, a draw constraint applied in a given round of a knockout tournament increases the probability of same nation matchups to approximately the same degree in each subsequent round, and it supports the teams of the strongest associations.

Consequently, the international character of the Europa League could have been improved further by prohibiting matches between teams from the same country even in the Round of 16 draw.

With the start of a new competition called UEFA Europa Conference League---the third tier of European club football---from the 2021/22 season, the Europa League has also been reformed to contain knockout round play-offs with 16 teams (contested by the eight group runners-up and the eight third-placed teams from the Champions League group stage) and a knockout stage starting with the Round of 16 (contested by the eight group winners and the eight winners of the play-offs). Analogously, the Europa Conference League group stage is followed by knockout round play-offs with 16 teams (contested by the eight group runners-up and the eight third-placed teams from the Europa League group stage) and a knockout stage starting with the Round of 16 (contested by the eight group winners and the eight winners of the play-offs).
Since there exists an association constraint in both the Europa League \citep[Articles~17 and 18]{UEFA2021d} and the Europe Conference League \citep[Articles~17 and 18]{UEFA2021e} draws of the knockout round play-offs and the Round of 16, UEFA has essentially implemented our recommendation above.

There are several ways to continue our research. Even though the complex interactions between competitors from many countries are difficult to handle with analytical tools, the mathematical results might be developed. The simulation model can be refined with respect to the winning probabilities. Finally, this study has ignored the problem that the restrictions make the draw unevenly distributed, which is unfair \citep{KlossnerBecker2013}. However, these distortions can be mitigated only by slacking the constraints \citep{BoczonWilson2022}.

Last but not least, it remains to be seen how the conclusions of our work can be exported to other tournaments. UEFA uses the same association constraint in the Champions League Round of 16 draw. However, this restriction is more effective in the Champions League compared to the Europa League as clubs from the top European leagues dominate the knockout stage of the former tournament to a higher degree and these national associations have usually more teams playing in the Champions League. As an illustration, consider the 2022/23 Champions League Round of 16 draw, where Bayern Munich was the only German group winner and Liverpool was the only English runner-up. Due to the presence of three English group winners and three German runners-up, the association constraint more than doubled the probability of a match between Bayern Munich and Liverpool from $1/7 \approx 14.29$\% to 37.12\% \citep{Guyon2022c}. Therefore, the effects of the association constraint are probably substantially higher in the Champions League. Hence, as we have argued at the end of Section~\ref{Sec2}, imposing the association constraint in the quarterfinals of the Champions League cannot be recommended.

Regarding different draw restrictions, we do not know that constraints other than prohibited clashes are used in knockout contests. Nonetheless, this possibility cannot be excluded since the draws of tournaments with round-robin groups sometimes contain further restrictions: at most two European nations can play in a FIFA World Cup group \citep{Csato2022e, FIFA2022a, Guyon2015a}, and the criteria used in the draw of the European Qualifiers for the 2022 FIFA World Cup \citep{Csato2022f, Csato2022a, UEFA2020c} are even more complex. In our opinion, it would be hazardous to estimate the effects of these constraints on the basis of the computations above; but the presented results can persuade the decision-makers to allocate resources to build a customised simulation model for this purpose.

\section*{Acknowledgements}
\addcontentsline{toc}{section}{Acknowledgements}
\noindent
This paper could not have been written without \emph{my father} (also called \emph{L\'aszl\'o Csat\'o}), who has primarily coded the simulations in Python. \\
\emph{D\'ora Gr\'eta Petr\'oczy} and six anonymous reviewers provided valuable comments and suggestions on earlier drafts. \\
We are indebted to the \href{https://en.wikipedia.org/wiki/Wikipedia_community}{Wikipedia community} for summarising important details of the sports competitions discussed in the paper.

\bibliographystyle{apalike}
\bibliography{All_references}

\begin{thebibliography}{}

\bibitem[Arlegi, 2022]{Arlegi2022}
Arlegi, R. (2022).
\newblock How can an elimination tournament favor a weaker player?
\newblock {\em International Transactions in Operational Research},
  29(4):2250--2262.

\bibitem[Arlegi and Dimitrov, 2020]{ArlegiDimitrov2020}
Arlegi, R. and Dimitrov, D. (2020).
\newblock Fair elimination-type competitions.
\newblock {\em European Journal of Operational Research}, 287(2):528--535.

\bibitem[Boczo{\'n} and Wilson, 2018]{BoczonWilson2018}
Boczo{\'n}, M. and Wilson, A.~J. (2018).
\newblock Goals, constraints, and public assignment: {A} field study of the
  {UEFA} {C}hampions {L}eague.
\newblock Technical Report 18/016, University of Pittsburgh, Kenneth
  P.~Dietrich School of Arts and Sciences, Department of Economics.
\newblock
  \url{https://www.econ.pitt.edu/sites/default/files/working_papers/Working%20Paper.18.16.pdf}.

\bibitem[Boczo{\'n} and Wilson, 2022]{BoczonWilson2022}
Boczo{\'n}, M. and Wilson, A.~J. (2022).
\newblock Goals, constraints, and transparently fair assignments: A field study
  of randomization design in the {UEFA} {C}hampions {L}eague.
\newblock {\em Management Science}, in press.
\newblock {DOI}:
  \href{https://doi.org/10.1287/mnsc.2022.4528}{10.1287/mnsc.2022.4528}.

\bibitem[Cea et~al., 2020]{CeaDuranGuajardoSureSiebertZamorano2020}
Cea, S., Dur{\'a}n, G., Guajardo, M., Saur{\'e}, D., Siebert, J., and Zamorano,
  G. (2020).
\newblock An analytics approach to the {FIFA} ranking procedure and the {W}orld
  {C}up final draw.
\newblock {\em Annals of Operations Research}, 286(1-2):119--146.

\bibitem[Csat\'o, 2021a]{Csato2021b}
Csat\'o, L. (2021a).
\newblock A simulation comparison of tournament designs for the {W}orld {M}en's
  {H}andball {C}hampionships.
\newblock {\em International Transactions in Operational Research},
  28(5):2377--2401.

\bibitem[Csat\'o, 2021b]{Csato2021a}
Csat\'o, L. (2021b).
\newblock {\em Tournament Design: How Operations Research Can Improve Sports
  Rules}.
\newblock Palgrave Pivots in Sports Economics. Palgrave Macmillan, Cham,
  Switzerland.

\bibitem[Csat\'o, 2022a]{Csato2022f}
Csat\'o, L. (2022a).
\newblock Fairer group draw for sports tournaments with restrictions.
\newblock Manuscript. {DOI}:
  \href{https://doi.org/10.48550/arXiv.2109.13785}{10.48550/arXiv.2109.13785}.

\bibitem[Csat\'o, 2022b]{Csato2022e}
Csat\'o, L. (2022b).
\newblock Group draw with unknown qualified teams: A lesson from 2022 {FIFA}
  {W}orld {C}up.
\newblock {\em International Journal of Sports Science \& Coaching}, in press.
\newblock {DOI}:
  \href{https://doi.org/10.1177/17479541221108799}{10.1177/17479541221108799}.

\bibitem[Csat\'o, 2022c]{Csato2022a}
Csat\'o, L. (2022c).
\newblock Quantifying incentive (in)compatibility: {A} case study from sports.
\newblock {\em European Journal of Operational Research}, 302(2):717--726.

\bibitem[FIFA, 2022]{FIFA2022a}
FIFA (2022).
\newblock {\em {D}raw procedures. {FIFA} {W}orld {C}up {Q}atar
  2022\textsuperscript{{TM}}}.
\newblock
  \url{https://digitalhub.fifa.com/m/2ef762dcf5f577c6/original/Portrait-Master-Template.pdf}.

\bibitem[Goossens et~al., 2012]{GoossensBelienSpieksma2012}
Goossens, D.~R., Beli{\"e}n, J., and Spieksma, F.~C.~R. (2012).
\newblock Comparing league formats with respect to match importance in
  {B}elgian football.
\newblock {\em Annals of Operations Research}, 194(1):223--240.

\bibitem[Guyon, 2014]{Guyon2014c}
Guyon, J. (2014).
\newblock A {B}etter {W}ay to {R}ank {S}occer {T}eams in a {F}airer {W}orld
  {C}up.
\newblock {\em The New York Times}.
\newblock 13 June.
  \url{https://www.nytimes.com/2014/06/14/upshot/a-better-way-to-rank-soccer-teams-in-a-fairer-world-cup.html}.

\bibitem[Guyon, 2015]{Guyon2015a}
Guyon, J. (2015).
\newblock Rethinking the {FIFA} {W}orld {C}up\textsuperscript{{TM}} final draw.
\newblock {\em Journal of Quantitative Analysis in Sports}, 11(3):169--182.

\bibitem[Guyon, 2022a]{Guyon2022a}
Guyon, J. (2022a).
\newblock ``{C}hoose your opponent'': A new knockout design for hybrid
  tournaments.
\newblock {\em Journal of Sports Analytics}, 8(1):9--29.

\bibitem[Guyon, 2022b]{Guyon2022c}
Guyon, J. (2022b).
\newblock Ligue des champions : le {B}ayern {M}unich, adversaire le plus
  probable du {PSG} en huiti\`emes de finale.
\newblock {\em Le Monde}.
\newblock 7 November.
  \url{https://www.lemonde.fr/sport/article/2022/11/07/ligue-des-champions-le-bayern-munich-adversaire-le-plus-probable-du-psg-en-huitiemes-de-finale_6148779_3242.html}.

\bibitem[Horen and Riezman, 1985]{HorenRiezman1985}
Horen, J. and Riezman, R. (1985).
\newblock Comparing draws for single elimination tournaments.
\newblock {\em Operations Research}, 33(2):249--262.

\bibitem[Hwang, 1982]{Hwang1982}
Hwang, F.~K. (1982).
\newblock New concepts in seeding knockout tournaments.
\newblock {\em The American Mathematical Monthly}, 89(4):235--239.

\bibitem[Jones, 1990]{Jones1990}
Jones, M.~C. (1990).
\newblock The {W}orld {C}up draw's flaws.
\newblock {\em The Mathematical Gazette}, 74(470):335--338.

\bibitem[Kendall and Lenten, 2017]{KendallLenten2017}
Kendall, G. and Lenten, L.~J.~A. (2017).
\newblock When sports rules go awry.
\newblock {\em European Journal of Operational Research}, 257(2):377--394.

\bibitem[Kiesl, 2013]{Kiesl2013}
Kiesl, H. (2013).
\newblock Match me if you can. {M}athematische {G}edanken zur
  {C}hampions-{L}eague-{A}chtelfinalauslosung.
\newblock {\em Mitteilungen der Deutschen Mathematiker-Vereinigung},
  21(2):84--88.

\bibitem[Kl{\"o}{\ss}ner and Becker, 2013]{KlossnerBecker2013}
Kl{\"o}{\ss}ner, S. and Becker, M. (2013).
\newblock Odd odds: The {UEFA} {C}hampions {L}eague {R}ound of 16 draw.
\newblock {\em Journal of Quantitative Analysis in Sports}, 9(3):249--270.

\bibitem[Laliena and L{\'o}pez, 2019]{LalienaLopez2019}
Laliena, P. and L{\'o}pez, F.~J. (2019).
\newblock Fair draws for group rounds in sport tournaments.
\newblock {\em International Transactions in Operational Research},
  26(2):439--457.

\bibitem[Lasek and Gagolewski, 2018]{LasekGagolewski2018}
Lasek, J. and Gagolewski, M. (2018).
\newblock The efficacy of league formats in ranking teams.
\newblock {\em Statistical Modelling}, 18(5-6):411--435.

\bibitem[Lenten and Kendall, 2021]{LentenKendall2021}
Lenten, L.~J.~A. and Kendall, G. (2021).
\newblock Scholarly sports: Influence of social science academe on sports rules
  and policy.
\newblock {\em Journal of the Operational Research Society}, in press.
\newblock {DOI}:
  \href{https://doi.org/10.1080/01605682.2021.2000896}{10.1080/01605682.2021.2000896}.

\bibitem[Rathgeber and Rathgeber, 2007]{RathgeberRathgeber2007}
Rathgeber, A. and Rathgeber, H. (2007).
\newblock Why {G}ermany was supposed to be drawn in the group of death and why
  it escaped.
\newblock {\em Chance}, 20(2):22--24.

\bibitem[Scarf et~al., 2009]{ScarfYusofBilbao2009}
Scarf, P., Yusof, M.~M., and Bilbao, M. (2009).
\newblock A numerical study of designs for sporting contests.
\newblock {\em European Journal of Operational Research}, 198(1):190--198.

\bibitem[Schwenk, 2000]{Schwenk2000}
Schwenk, A.~J. (2000).
\newblock What is the correct way to seed a knockout tournament?
\newblock {\em The American Mathematical Monthly}, 107(2):140--150.

\bibitem[UEFA, 2015]{UEFA2015a}
UEFA (2015).
\newblock {\em Regulations of the UEFA Champions League 2015-18 Cycle. 2015/16
  Season}.
\newblock
  \url{http://web-old.archive.org/web/20201018192502/https://www.uefa.com/MultimediaFiles/Download/Regulations/uefaorg/Regulations/02/23/57/51/2235751_DOWNLOAD.pdf}.

\bibitem[UEFA, 2019]{UEFA2019b}
UEFA (2019).
\newblock 2019/20 {UEFA} club competitions revenue distribution system.
\newblock 11 July.
  \url{https://www.uefa.com/insideuefa/stakeholders/clubs/news/newsid=2616265.html}.

\bibitem[UEFA, 2020a]{UEFA2020c}
UEFA (2020a).
\newblock {FIFA} {W}orld {C}up 2022 qualifying draw procedure.
\newblock
  \url{https://www.uefa.com/MultimediaFiles/Download/competitions/WorldCup/02/64/22/19/2642219_DOWNLOAD.pdf}.

\bibitem[UEFA, 2020b]{UEFA2020b}
UEFA (2020b).
\newblock {\em Regulations of the UEFA Europa League 2018-21 Cycle. 2020/21
  Season}.
\newblock
  \url{http://web.archive.org/web/20210517095618/https://documents.uefa.com/r/Regulations-of-the-UEFA-Europa-League-2020/21-Online}.

\bibitem[UEFA, 2021a]{UEFA2021e}
UEFA (2021a).
\newblock {\em Regulations of the UEFA Europa Conference League 2021-24 Cycle.
  2021/22 Season}.
\newblock
  \url{https://web.archive.org/web/20220208043024/https://documents.uefa.com/r/Regulations-of-the-UEFA-Europa-Conference-League-2021/22-Online}.

\bibitem[UEFA, 2021b]{UEFA2021d}
UEFA (2021b).
\newblock {\em Regulations of the UEFA Europa League 2021-24 Cycle. 2021/22
  Season}.
\newblock
  \url{https://web.archive.org/web/20220207235133/https://documents.uefa.com/r/Regulations-of-the-UEFA-Europa-League-2021/22-Online}.

\bibitem[Wright, 2014]{Wright2014}
Wright, M. (2014).
\newblock {OR} analysis of sporting rules -- {A} survey.
\newblock {\em European Journal of Operational Research}, 232(1):1--8.

\end{thebibliography}

\end{document}